\documentclass[reqno]{amsart}
\usepackage[colorlinks=true,linkcolor=blue]{hyperref}    
\usepackage{amssymb}
\usepackage{amsfonts}
\usepackage{amsthm}
\usepackage{amsmath}
\usepackage{mathtext}
\usepackage{latexsym}
\usepackage{color}
\DeclareMathSymbol{\not}{\mathrel}{symbols}{"36}
\newcommand{\point}{\par\noindent$\bullet$ \ }
\newcommand{\dsize}{\textstyle}
\newcommand{\R}{{\mathbb R}}
\newcommand{\Z}{{\mathbb Z}}
\newcommand{\N}{{\mathbb N}}

\newcommand{\C}{{\mathbb C}}
\newcommand{\QV}{{\mathcal Q}}
\newcommand{\re}{{\rm Re}\,}
\newcommand{\im}{{\rm Im}\,}
\newcommand{\Ai}{{\mathrm Ai}\,}
\theoremstyle{plain}
\newtheorem{Th}{Theorem}[section]
\newtheorem{Le}{Lemma}[section]
\newtheorem{Pro}{Proposition}[section]
\newtheorem{Cor}{Corollary}[section]
\theoremstyle{definition}
\newtheorem{Rem}{Remark}[section]

\begin{document}
%
\title{Adiabatic evolution generated by a one-dimensional 
Schr\"odinger operator with decreasing number of eigenvalues} 
\author{Alexander Fedotov}
\address{Saint Petersburg State University, 7/9 Universitetskaya nab., 
St.Petersburg, 199034 Russia}
\email{a.fedotov@spbu.ru}
\begin{abstract}
We study a one-dimensional non-stationary Schr\"o\-dinger equation 
with a potential slowly depending on time. The corresponding 
stationary operator depends on time as on a parameter. It has a 
finite number of negative eigenvalues and absolutely continuous 
spectrum filling $[0,+\infty)$.  The eigenvalues move with time 
to the edge of the continuous spectrum and, having reached it, 
disappear one after another. We describe the asymptotic behavior of a solution 
close at some moment to an eigenfunction of the stationary operator, and,
in particular, the phenomena occurring when the corresponding  eigenvalue 
approaches the absolutely continuous spectrum and disappears.
\end{abstract}
\thanks{The present work was supported by the Russian foundation of basic 
research under grant  14-01-00760-a and by the Saint Petersburg State University
under grant 11.38.263.2014.}
\maketitle
\section{Introduction}
As $\varepsilon\to 0$, we study solutions to the  Schr\"odinger equation
\begin{equation}\label{non-stat-eq}
i\frac{\partial\psi}{\partial t}=
-\frac{\partial^2\psi}{\partial x^2}+v(x,\varepsilon t) \psi. 
\quad x>0,\qquad \left.\psi\right|_{x=0}=0.
\end{equation}
When  $\varepsilon$ is small, one says that~\eqref{non-stat-eq} 
describes the {\it adiabatic evolution} in $L_2(\R_+)$ generated  by 
the {\it stationary} operator $H(\varepsilon t)=
-\frac{\partial^2}{\partial x^2}+v(x,\varepsilon t)$ with the Dirichlet 
boundary condition at zero. Note that this operator depends on time 
as on a parameter. 

The adiabatic evolution generated by differential operators 
is a classical object  of study in mathematical physics and, 
in particular, in quantum mechanics, see the review in~\cite{Av-El:99}. 
The formulation of one of the standard problems assumes that 
the spectrum of the corresponding stationary operator is discrete. 
For this problem, an important role is played by the  solutions that, 
at some moment in time, are close to the eigenfunctions of the 
stationary operator. In the case of~\eqref{non-stat-eq}, instead 
of such solutions, physicists often use formal asymptotic series 
of the form
\begin{equation}\label{adia} 
e^{-\frac{i}\varepsilon\int\limits_{\tau_0}^{\varepsilon t}E_n(\tau)\,d\tau}\,
\sum_{m=0}^{\infty}\varepsilon^m\,\psi_{n,m}(x,\varepsilon t),\quad \varepsilon\to0,
\end{equation}
where $E_n(\tau)$ is an eigenvalue of  $H(\tau)$, 
$\psi_{n,0}(\cdot,\tau)$ is the corresponding eigenfunction,
and  $\tau_0$ is a fixed number.     

We consider a model problem with an  operator $H(\varepsilon t)$  
having a finite number of negative eigenvalues and absolutely continuous 
spectrum filling $\R_+$. The eigenvalues move with time toward the 
continuous spectrum and, having reached it, disappear one after another. 
We study a solution having the asymptotics of the form~\eqref{adia} as 
long as the $n$-th eigenvalue exists. There is a whole bunch of effects
occurring as a result of ``absorption'' of the eigenvalue by the continuous 
spectrum. In this paper, we begin to analyze them. Some of our results 
were announced in the note~\cite{F-S:2016}. 
\section{Main results}
Here, we describe the potential $v$ that we consider, the solution 
to~\eqref{non-stat-eq} that we study, and the asymptotics of this solutions 
that we get in this paper.
\subsection{The model we consider}
We assume that $0<\varepsilon<1$ and  study the case where
\begin{equation}\label{potential}
v(x,\tau)=\begin{cases} -1&\text{if \ } 0\le x\le 1-\tau,\\
\ 0 & \text{otherwise}.\end{cases}
\end{equation}
This allows to construct solutions to~\eqref{non-stat-eq} using ideas similar
to ideas  of the Sommer\-feld-Malyuzhinets method developed to study scattering 
of waves on wedge-shaped domains~\cite{B-L-G:2008}: by means of a suitable 
integral transform, one expresses solutions of~\eqref{non-stat-eq} in terms 
of solutions of a difference equation on the complex plane. In our case, 
the translation parameter in this equation equals $\varepsilon$. 
The smallness of  $\varepsilon$ enables to carry out an effective analysis.

In this paper, to keep its size reasonable,  
we mostly restrict ourselves to the analysis of solutions inside 
the potential well. We also assume that $\varepsilon t\le 1$. 
\subsection{The solution we study}\label{ss:Ansatz}
In section~\ref{gen-sol}, we construct a solution $\Psi$ having the following
simple wave interpretation. It depends on a parameter $p\in\R$, and
inside the wedge $W=\{(x,t)\in\R^2\,:\,0\le x\le 1-\varepsilon t\}$, \ 
$\Psi$ is a linear combination of the plane wave $e^{-i(p^2-1)t+ipx}$ and all 
the plane waves that can be obtained from it by reflections from the 
boundaries of $W$. Outside the wedge, $\Psi$ can be considered 
as a linear combination of all the refracted waves. 
Describe briefly the construction of $\Psi$. We choose the Ansatz
\begin{equation}\label{sol-form}
\Psi(x,t,p)=\frac1{\sqrt{\pi}}
\begin{cases}\ e^{it}\sum_{k=p+\varepsilon l,\,l\in\Z} 
e^{-ik^2(t-1/\varepsilon)} sin(kx)\,R(k), \ \   
0\le x< 1-\varepsilon t,\\ \\
\ \sum_{k=p+\varepsilon l,\,l\in \Z} 
e^{-ip_1^2(k)(t-1/\varepsilon)+ip_1(k)x}\,T(k)\,R(k), \ \  
1-\varepsilon t\le x.
\end{cases}
\end{equation}
If the series in~\eqref{sol-form} converge sufficiently well,  $\Psi$  
satisfies~\eqref{non-stat-eq} and the Dirichlet boundary condition. The series 
appears to be convergent and $\Psi$ appears to be continuously differentiable 
in $x$ at $x=1-\varepsilon t$ if $R$ satisfies the equation
\begin{equation}\label{eq:R}
R(p+\varepsilon/2)=\rho(p)\,R(p-\varepsilon/2),\quad 
\rho(p)=\frac{Q(p)-p}{Q(p)+p},\quad Q(p)=\sqrt{p^2-1},
\end{equation}
$T$ and $p_1$ are defined  in terms of $Q$ by formulas~\eqref{eq:T} 
and~\eqref{K:formula}, and the branch of the function $Q$ in the definition 
of $\rho$ in~\eqref{eq:R} is chosen in a suitable way. 

We call $\Psi$ a {\it generating solution}. It  is $\varepsilon$-periodic 
in $p$. Its Fourier coefficients $\Psi_n$ with $n\ge1$ are the solutions that 
we study. They  satisfy the Dirichlet condition at $x=0$ and share with 
$\Psi$ all the regularity properties in $x$. In the potential well, 
\begin{equation} 
  \label{eq:Psi-n:up}
 \Psi_n(x,t)=\frac{e^{it}}{\sqrt{\varepsilon\pi}} 
\int_{-\infty}^\infty e^{i(p^2(1-\tau)-2\pi np)/\varepsilon} sin(px)\,R(p)\, dp,\quad
\tau=\varepsilon t.
\end{equation} 

Before analyzing $\Psi_n$,  we study solutions 
to~\eqref{eq:R} in Section~\ref{sec:as-R}.
\subsection{Properties of the operator $H(\tau)$}\label{ss:notations}
When describing the asymptotics of $\Psi_n$, we use the following 
elementary facts.  Consider the  operator $H(\tau)$  with the 
potential~\eqref{potential}. Let 
\begin{equation}
  \label{eq:tau_N}
  \tau_n=1-\pi(n-1/2),\quad n\in\N.
\end{equation}
The number of  negative eigenvalues of 
$H(\tau)$ equals $N\in\N$ if  $\tau_{N+1}<\tau<\tau_N$. If 
$\tau>\tau_1$, there are no eigenvalues. We use the notation  
\begin{equation}\label{eq:cn}
c_n=e^{\frac{i}\varepsilon\,(2\tau_n-3)+\frac{i\pi}4}.
\end{equation}
Let  $E_n(\tau)$ be the $n$th eigenvalue of $H(\tau)$. It can be represented in 
the form $E_n(\tau)=p_n^2(\tau)-1$ with $p_n\in(0,1)$ satisfying the equation 
\begin{equation}\label{eq:pn-up}
(1-\tau)\,p_n(\tau) +\arcsin(p_n(\tau))=\pi n.
\end{equation}
The eigenfunction corresponding to $E_n(\tau)$ is given by the formulas
\begin{equation}\label{eq:psi-n}
\psi_n(x,\tau)=\begin{cases}\sin(p_n(\tau) x) & 
\text{if \ } 0\le x\le 1-\tau,\\  
(-1)^{n+1}\,p_n(\tau)\,e^{-(x-1+\tau)\sqrt{1-p_n^2(\tau)}}& 
\text{if \ } x\ge 1-\tau.
\end{cases}
\end{equation} 
\subsection{Standard asymptotic behavior}
Fix  $n\in N$. In section Section~\ref{sec:standard}, 
we study the solution $\Psi_n(x,t)$ in the case when 
$E_n(\varepsilon t)$ exists and is bounded away from zero. 
In the potential well,  $\Psi_n$ has an asymptotic 
expansion of the form~\eqref{adia}:
\begin{Th} \label{th:norm-wave:up} 
Fix $T_1<T_2<\tau_n$ and $K>0$. For  $T_1\le\varepsilon t\le T_2$ and 
$0\le x\le 1-\varepsilon t$, as $\varepsilon\to 0$,
\begin{equation}
\label{eq:Psi_n-adiabatic:up}
  \Psi_n(x,t)= c_n\sqrt{\frac{d\ln p_n}{d\tau}(\varepsilon t)}\,
e^{-\frac{i}\varepsilon\int\limits_{\tau_n}^{\varepsilon t}E_n(\tau)\,d\tau}\,
\left(\sum_{k=0}^{K-1}\varepsilon^k\psi_{n,k}(x,\varepsilon t)
+O(\varepsilon^K)\right),
\end{equation}
where $\psi_{n,0}=\psi_n$, and $\psi_{n,k}$ are 
bounded. This representation is uniform 
in $x$ and $t$.
\end{Th}
Outside the potential well, $\Psi$ has two-scale asymptotic expansion.
Describe it. The equation
\begin{equation}
  \label{eq:pn:down}
(1-\tau)\tilde p_n+\arcsin \tilde p_n
-i\tilde p_n\xi/(2\sqrt{1-\tilde p_n^2})=\pi n,
\end{equation}
allows to define a continuous function 
$\tilde p_n: \{\tau<1,\;\xi> 0\}\mapsto{\mathcal Q}_1=\{p\in\C\,:\,
\re p,\im p>0\}$ 
such that $\tilde p_n(\tau,0)=p_n(\tau)$ when $\tau<\tau_n$. 
One has 
\begin{Th} \label{th:norm-wave:down}  Fix  $K\in\N$, $T_1<T_2<\tau_n$ 
and $X>0$.  Let  $T_1\le \varepsilon t\le T_2$ and
$1-\varepsilon t\le x\le X/\varepsilon$. Let
$\xi=\varepsilon (x-(1-\tau))$, $\tau=\varepsilon t$.
As $\varepsilon\to0$, 
\begin{gather}\label{as:Psi:down:1}
\Psi(x,t)=c_n
\sqrt{\frac{\partial\ln \tilde p_n}{\partial\tau}}
e^{-\frac{i}\varepsilon\int\limits_{\tau_n}^{\tau}E_n(\tau)\,d\tau}
\phi_n(\xi,\tau)
\left(\sum_{k=0}^{K-1}\varepsilon^k\alpha_{n,k}(\xi,\tau)+O(\varepsilon^K)\right),\\
\label{eq:Psi-n-adiabatic:down:1}
\phi_{n}(\xi,\tau)=(-1)^{n+1}\tilde p_n
e^{-\frac{1}\varepsilon\int_{0}^{\xi}\sqrt{1-\tilde p_n^2}\,d\xi-\frac{i\xi}2},\quad 
\tilde p_n=\tilde p_n(\tau,\xi),
\end{gather}  
the branch of the square root is such that
$\re \sqrt{1-{\tilde p}^2}>0$ for $\tilde p\in{\mathcal Q}_1$,
\ $\alpha_{n,0}\equiv1$, and  $\alpha_{n,k}(\tau,\xi)$ are bounded.  
This representation is uniform in $x$ and $t$.
\end{Th}
\noindent In this paper, we only outline the proof this theorem.
Note that (1) as $\re \sqrt{1-{\tilde p}^2}>0$, \  $\phi_n(\xi,\tau)$ decays as 
$\xi$ increases;  (2) if $\varepsilon^{1/2}x\to 0$, \  
$\phi_n(\varepsilon x, \tau)=\psi_n(x,\tau)(1+o(1))$.  

For large $\xi$, $\Psi_n$ is described by
\begin{Le}\label{le:large-xi-est}  
Fix $0<c<1$.  There is a $C>0$ such that, 
for  sufficiently large 
$\xi=\varepsilon (x-(1-\varepsilon t))$, one has
$  |\Psi_n(x,t)|\le Ce^{-c\xi/\varepsilon}$.
\end{Le} 
\subsection{Destruction of the standard asymptotic behavior, 
$\varepsilon t \le\tau_n$} 
As $\varepsilon t$ grows and approaches $\tau_n$, the eigenvalue 
$E_n(\tau)$ approaches the edge of the absolutely continuous spectrum, and
the asymptotic behavior of $\Psi_n$ changes. Set
\begin{equation}\label{eq:F}
F(z)= \sqrt{\pi}\,e^{-\frac{2z^3}3-\frac{i\pi}{12}}\,
\left(z\,\text{Ai}\,(z^2)-\text{Ai}'\,(z^2)\right),
\end{equation} 
where $\rm Ai$ is the Airy function. 
One has  $F(e^{\frac{i\pi}6}z)=z^{1/2}\,e^{-4iz^3/3}\,(1+o(1))$
as $z\to+\infty$, and $F(e^{\frac{i\pi}6}z)=(-i/8+o(1))\;(-z)^{-5/2}$
as $z\to-\infty$. 
In section~\ref{sec:st-behavior-destruction}, we  prove 
\begin{Th} \label{th:tau-smaller-tau-n}
Fix a sufficiently small $\delta>0$. 
Let $\tau_n-\delta\le \varepsilon t\le \tau_n$ and 
$0\le x\le 1-\varepsilon t$. Then, as $\varepsilon\to 0$,
\begin{equation}
\label{psi:n:tau-less-tau-n}
\Psi_n(x,t)=c_n\sqrt{\frac1{Z_n}\frac{\partial \ln p_n}{\partial\tau}}
\psi_{n}F(e^{\frac{i\pi}6}Z_n)
+O\left(\varepsilon^{\frac23}(1+|Z_n|^{\frac12})\right), 
\end{equation}
where $Z_n=\left(\frac{3}{4\varepsilon}
\int_{\tau_n}^{\varepsilon t}E_n\,d\tau\right)^{\frac13}$,
$p_n=p_n(x,\varepsilon t)$, \ $\psi_{n}=\psi_{n}(x,\varepsilon t)$, 
and $E_n=E_n(\tau)$. The asymptotic representation is uniform 
in $x$ and $t$. 
\end{Th}
If $\varepsilon\to 0$ and $\tau_n-\varepsilon t$ stays of the 
order of $1$, then $Z_n\to+\infty$, and the leading term 
in~\eqref{psi:n:tau-less-tau-n} turns into the leading term 
from~\eqref{eq:Psi_n-adiabatic:up}.

In  this paper, we do not study  $\Psi$ outside the potential well 
for $\varepsilon t\sim \tau_n$. We mention only 
that, in the domain where the expression 
$\varepsilon^{-2/3}(\varepsilon t-\tau_n)^2+\varepsilon^{1/3}x$
is of order of one,  $\Psi_n$ is described in terms of Airy functions,
and, for larger  $x$, it exponentially decays  as  $x$ increases.
\subsection{Aftermath}
If
\begin{equation}\label{eq:tau-ge-tau-n}
0\le x\le 1-\tau,\qquad \tau_n\le\tau\le 1,\qquad \tau= \varepsilon t,
\end{equation}
then,  up to some error terms, the solution $\Psi_n$ appears to be 
the sum of three terms ``responsible'' for three different phenomena. 
First, we describe these terms, and then, formulate a theorem. 
\\
{\bf ``Transition'' term.} \ The first term is described by the formula 
\begin{equation}\label{as:T}
{\mathcal T}_0(x,\varepsilon t)=
(4\varepsilon)^{\frac16}c_n\sin x\, F\left(e^{\frac{i\pi}6}z_n(\varepsilon t)\right),
\quad
z_n(\tau)=\frac{\tau_n-\tau}{(4\varepsilon)^{1/3}}.
\end{equation}
One can see  that, when $\varepsilon t-\tau_n$ is of the order of 
$\varepsilon^{1/3}$, the leading term in~\eqref{psi:n:tau-less-tau-n} turns 
into ${\mathcal T}_0$.  On the other hand, as 
$(\varepsilon t-\tau_n)/\varepsilon^{1/3}\to+\infty$, 
the function $F$ in~\eqref{as:T} can be ``replaced'' by its 
asymptotics. This leads to an asymptotics of ${\mathcal T}_0$ with the leading term
$\frac{-ic_n\varepsilon}{2(\tau-\tau_n)^{5/2}}\,\sin x$, 
and ${\mathcal T}_0$ becomes of the order of  $\varepsilon$.
\\
{\bf ``Resonance'' term.} 
Put
\begin{equation}
  \label{eq:a(lambda)}
  a(z)=\int_0^\infty e^{-\frac{u^3}3+iz u^2}u\,du.
\end{equation}
The function $a$ is a close relative of the Airy function.
By means of the method of steepest descents, one checks that
\begin{equation}\label{as:a(lambda)}
    a(z)=\frac{i}{2z}+O\left(z^{-\frac52}\right),\quad 
z\to\pm\infty.
\end{equation}
This representation can be differentiated infinitely many times.
The second term is given by the formula
\begin{equation}
  \label{as:R}
{\mathcal R}_0(x,\varepsilon t)=\frac{c_n\sin x}{\pi^{\frac32}}
\sum_{k=0}^\infty f_k\,\left(\left(\frac\varepsilon2\right)^{\frac23}\,
a(z_{n-k}(\tau))-
\frac{i(1-\tau)\varepsilon}{16}\,a''(z_{n-k}(\tau))\right),
\end{equation}
where $\tau=\varepsilon t$, and 
\begin{equation}
  \label{eq:fk}
  f_k=\frac{(-1)^k}{k^{\frac32}},\quad k>0,\qquad
  f_0=-\sum_{k=1}^\infty f_k.
\end{equation}
Let us discuss the term ${\mathcal R}_0$. 
If $\varepsilon t$ is outside a fixed neighborhood
of $\tau_1,\tau_2\dots\tau_n$, then, in view of~\eqref{as:a(lambda)},
$${\mathcal R}_0(x,\varepsilon t)=\frac{i\varepsilon c_n\sin x}{2\pi^{\frac32}}
\sum_{k=0}^\infty \frac{f_k}{\tau_{n-k}-\varepsilon t}+O(\varepsilon^{\frac32})=
O(\varepsilon).$$
On the other hand,  assume that, for an integer $1\le N\le n-1$,
one has $\tau_N-\tau\asymp \varepsilon^{1/3}$. Then,~\eqref{as:R} turns into 
$${\mathcal R}_0(x,\varepsilon t)=\frac{c_n\sin x}{\pi^{\frac32}}
\left(\frac\varepsilon2\right)^{\frac23}
f_{n-N}\,\,a(z_{N}(\tau))+O(\varepsilon).$$
So, ${\mathcal R}_0$ becomes relatively large near the moments 
$t=\tau_l/\varepsilon$, $l=1,2, \dots n-1$, i.e., the moments of ``death'' 
of the eigenvalues  of the stationary operator $H_\varepsilon$. Note 
that these  can be interpreted as the moments of birth of its 
resonances. 
\\
{\bf Between the moments $\tau_n$, $\tau_{n-1}$ $\dots$} \ 
If $\varepsilon t-\tau_n\le \varepsilon^{1/3}$, we set 
${\mathcal G}_0(x,\varepsilon t)=0$ , and 
otherwise we define ${\mathcal G}_0$ by the formula
\begin{equation}
  \label{as:F:tau-ge-tau-n:2}
  {\mathcal G}_0(x,\varepsilon t)=
i c_n\,\sqrt{\frac2\pi}\ \frac{\varepsilon\sin x}{\varepsilon t-\tau_n}\,
\re \int_0^\infty e^{-2s(\varepsilon t-\tau_n)}
\left(e^{i\pi/4}\zeta(is)+2\sqrt{s}\right)\,ds,
\end{equation}
where $\zeta$ is analytic  in $\C\setminus[1/2,\infty)$ and given 
there by~\eqref{eq:zeta}. 
Thanks to~\eqref{eq:zeta-as}, the integral in~\eqref{as:F:tau-ge-tau-n:2}
converges uniformly in $\varepsilon t\ge \tau_n$. 

The ${\mathcal G}_0$ is the third term that contributes to the leading term 
of the asymptotics  of $\Psi_n$.  If  the distance from $\varepsilon t$ to 
$\tau_1,\tau_2 \dots \tau_n$ is bounded away from zero by a fixed constant, 
then, as $\varepsilon\to 0$, all the three terms ${\mathcal T}_0$, 
${\mathcal R}_0$ and ${\mathcal G}_0$ are of the order of $\varepsilon$.

In Section~\ref{sec:tau-greater-tau-n} we prove 
\begin{Th}\label{th:final} 
Let $x$ and $t$ satisfy~\eqref{eq:tau-ge-tau-n}.
As $\varepsilon\to0$, 
\begin{align*}
  \Psi_n(x,t)={\mathcal T}_0(x,\varepsilon t)&\,(1+O(\varepsilon^{1/2}))+
{\mathcal R}_0(x,\varepsilon t)+{\mathcal G}_0(x,\varepsilon t)+\\
&+O(\varepsilon^{7/6})+O(\varepsilon^{2/3}/
(1+|z_n(\varepsilon t)|)^{5/2}).
\end{align*}
\end{Th} 
\subsection{Final remarks}
In the underwater acoustics, in  the course of investigation of the 
sound propagation in a narrow water wedge near a sea shore, problems similar 
to ours were  studied (non-rigorously). For example, in~\cite{Pi:83}, the author 
got an elegant partial differential equation for the sound field in the case 
similar 
to the case where  $\varepsilon t\sim\tau_n$.  This enabled him to describe the 
the leading order approximation of the field in terms of Airy functions. 
It looks like that the physicists have not discovered the effects described 
by the term ${\mathcal R}_0$, neither have they found the  function $\zeta$ 
determining the term ${\mathcal G}_0$.

We hope that our results can be generalized to the case where the potential
$v$ in equation~\eqref{non-stat-eq} is non-positive, quickly vanishes as  
$x\to\infty$, and the eigenvalues  $E(\tau)$ behave for $\tau\sim\tau_n$ as 
in the model problem.  In particular, we expect that the leading 
terms of the asymptotics of $\Psi_n$ for $\varepsilon t\ge\tau_n$,
can be obtained from ${\mathcal T}_0$, ${\mathcal R}_0$ and ${\mathcal G}_0$
by replacing $\sin x$ with a solution of the equation 
$-\psi_{xx}''+v(x,\tau)\psi=0$ satisfying the Dirichlet condition at $x=0$.
  
The results we got admit the following physical interpretation.
For $t<\tau_n/\varepsilon$, the quantum particle described by the wave 
function  $\Psi_n$ is in the state with the energy  $E_n(\varepsilon t)$ 
and is localized in the potential well  $0\le x\le \varepsilon t$. 
When $t$ approaches $\tau_n/\varepsilon$, the moment of ``death'' of
the energy level $E_n(\varepsilon t)$, the probability to find the quantum 
particle in the potential well decreases and becomes small when 
$t>\tau_n/\varepsilon$. The energy of the particle moves into the continuous 
spectrum. After that, when  $t$ is close to $\tau_{n-1}/\varepsilon$, 
$\tau_{n-2}/\varepsilon$, $\dots$, i.e.,  to the moments of birth
of the resonances, thanks to tunneling effects, the probability to find the 
particle in the potential well again becomes noticeable.
\section{Generating solution}\label{gen-sol}
Here, we construct the  generating solution as described in 
Section~\ref{ss:Ansatz}. 
\subsection{Relations for the functions $R$, $T$ and $p_1$}
\label{ss:formal-construction-of-Psi}
Assume that the series in formulas~\eqref{sol-form}  converge sufficiently 
well.  The function $\Psi$ defined by these formulas is continuous in $x$ 
at $x=1-\varepsilon t$ for all $t< 1/\varepsilon$, if, for all $p\in\R$ and  
for all $x\in\R$,
\begin{equation}
\label{eq:continuity}
\begin{split}
e^{i(1-x)/\varepsilon} \left(e^{ip^2 x/\varepsilon+ipx}\,R(p)\right.-&\left. 
e^{i(p+\varepsilon)^2 x/\varepsilon-i(p+\varepsilon)x}
\,R(p+\varepsilon)\right)=\\
&=2i T(p)R(p)e^{ip_1(p)^2 x/\varepsilon+ip_1(p) x}.
\end{split}
\end{equation}
As $ p^2+\varepsilon p= 
(p+\varepsilon)^2-\varepsilon (p+\varepsilon)$,~\eqref{eq:continuity} implies that
\begin{gather}
\label{eq:K}
-1+p^2+\varepsilon p=p_1(p)^2 +\varepsilon p_1(p), \\
\label{eq:R-T-1}
R(p)-R(p+\varepsilon)=2i e^{-i/\varepsilon}T(p)R(p).
\end{gather}
Furthermore, $\partial\Psi/\partial x$ becomes continuous in $x$
at $x=1-\varepsilon t$ for all $t< 1/\varepsilon$, if, for all 
$p\in\R$ and  for all $x\in\R$,
\begin{equation}
\label{eq:continuity1}
\begin{split}
e^{i(1-x)/\varepsilon} \left(e^{ip^2 x/\varepsilon+ipx}\,pR(p)\right.+&\left. 
e^{i(p+\varepsilon)^2 x/\varepsilon-i(p+\varepsilon)x}
\,(p+\varepsilon)R(p+\varepsilon)\right)=\\
&=2i e^{ip_1(p)^2 x/\varepsilon+ip_1(p) x}p_1(p)T(p)R(p)
\end{split}
\end{equation}
which leads to the relation
\begin{equation}
\label{eq:R-T-2}
pR(p)+(p+\varepsilon)R(p+\varepsilon)=2i e^{-i/\varepsilon}\,p_1(p)\,T(p)\,R(p).
\end{equation}

Let us discuss~\eqref{eq:K},~\eqref{eq:R-T-1} and~\eqref{eq:R-T-2}.
Equation~\eqref{eq:K} implies that
\begin{equation}\label{K:formula}
p_1(p)=-\frac{\varepsilon}2+Q\left(p+\frac\varepsilon2\right),\qquad 
Q(p)=\sqrt{p^2-1}.
\end{equation}
From~\eqref{eq:R-T-1} and~\eqref{eq:R-T-2}, we deduce the equation
for $R$  from~\eqref{eq:R} and the formula
\begin{equation}
\label{eq:T}
T(p)=-i\left.\frac{pe^{i/\varepsilon}}{Q(p)+p}\right|_{p:=p+\varepsilon/2}.
\end{equation}
\subsection{Function $R$}\label{ss:R-function}
Here, we  construct a solution to equation~\eqref{eq:R} for complex $p$.
\subsubsection{A branch of the $Q$}\label{sssec:Q}
The branch points of $Q$ are the points $\pm1$. 
Let $\C_0=\C\setminus\{p\in\R\,:\, |p|\ge 1\}$. In $\C_0$, we fix 
the single-valued analytic branch $Q_0$ of $Q$ so that $Q_0(0)=i$. 
Note that 
\par\vspace{2mm}\noindent
\begin{equation}
  \label{eq:Q_0}
  Q_0(p)\in i\R_+\quad\text{if}\quad  -1<p<1,\quad\text{and}\quad
  Q_0(p+i0)\in \R_\pm\quad\text{if}\quad  \pm p>1.
\end{equation}
\par\vspace{2mm}\noindent
The function $Q_0$ is even and $\rho_0$,  the coefficient $\rho$ 
defined by~\eqref{eq:R} for $Q=Q_0$, satisfies the relation 
$\rho_0(-p)=1/\rho_0(p)$. 
One has 
\begin{equation}\label{est:rho}
 \rho_0(p+ i0)\asymp 1/p^2 \quad \text{as}\quad |p|\to\infty,\quad p\in\R.
\end{equation}
\subsubsection{Constructing a solution to~\eqref{eq:R}}
To describe a solution to~\eqref{eq:R}, we need some notations.
In $\C_0$, we define an analytic function by the formulas
\begin{equation}\label{eq:l-def}
l_0(p)=-i\ln \rho_0(p)=-i\ln\frac{Q_0(p)-p}{Q_0(p)+p},\quad l_0(0)=0.
\end{equation} 
The definition of $l_0$ implies that
\begin{equation}
  \label{eq:l0-arsin}
  l_0=2\arcsin (p),\quad -1\le p\le 1,
\end{equation}
and that
\begin{equation}\label{eq:l0-sym}
  l_0(-p)=-l_0(p),\quad l_0(\bar p)=\overline{l_0(p)}.
\end{equation}
We call a curve $\gamma\subset \C$ vertical if, along $\gamma$,  
$p$ is a piecewise $C^1$-function of  $\im p\in\R$, and $d p/d\im p$ 
is uniformly bounded.
For $p_0\in\C$, we denote by $\gamma(p_0)\subset \C$ a vertical curve 
containing $p_0$.
One has 
\begin{Pro}\label{pro:R-func}  For $p\in\C_0$, let 
\begin{equation}\label{for:R}
R_0(p)=\exp\left(\frac i\varepsilon \int_0^p L_0(p)\,dp\right),\qquad
L_0(p)=\frac{\pi}{2i\varepsilon}\int\limits_{\gamma(p)}
\frac{l_0(\zeta)\,d\zeta} { \cos^2 
\left(\frac {\pi(p-\zeta)}\varepsilon\right)}, 
\end{equation}
and $\gamma(p)\subset \C_0$. The functions $L_0$ and $R_0$  
are analytic in $\C_0$, and $R_0$ is continuous up to its boundary.
The $R_0$ is  a solution to the equation in~\eqref{eq:R}, and $L_0$ 
satisfies the relations
\begin{gather} \label{eq:L}
L_0(p+\varepsilon/2)-L_0(p-\varepsilon/2)=\varepsilon l_0'(p), \quad 
p\pm\varepsilon/2\in \C_0,\\
 \label{eq:L-prop}
  L_0(\bar p)=\overline{L_0(p)},\quad L_0(-p)=-L_0(p),\quad p\in\C_0.
\end{gather}
\end{Pro}
\begin{proof}
The analyticity of $L_0$ (and $R_0$) follows from the choice of 
$\gamma(p)$ and the estimate
\begin{equation*}
 |l_0(p)|\le C\ln (2+|p|),\quad p\in\C_0,
\end{equation*}
$C>0$ being a constant. Using the residue theorem, one checks 
that $L_0$ solves~\eqref{eq:L}. To prove~\eqref{eq:L-prop}, one checks 
by means of~\eqref{eq:l0-sym} that, for $p\in i\R$, \ $L_0(p)\in i\R$ and 
$L_0(-p)=-L_0(p)$.

Let us prove that $R_0$ satisfies~\eqref{eq:R}.
To simplify the notations,  we write below $l$ and $L$ instead of $l_0$ and  
$L_0$. For $p\pm \varepsilon/2\in\C_0$, one has
\begin{gather*}
\int\limits_0^{p+\frac\varepsilon{2}}L(p)\,dp-
\int\limits_0^{p-\frac\varepsilon{2}}L(p)\,dp=
  \int\limits_{-\frac\varepsilon{2}}^p L\left(p+\frac\varepsilon{2}\right)\,dp-
  \int\limits_{\frac\varepsilon{2}}^p L\left(p-\frac\varepsilon{2}\right)\,dp=\\
  \int\limits_{-\frac\varepsilon{2}}^0 L\left(p+\frac\varepsilon{2}\right)dp-
  \int\limits_{\frac\varepsilon{2}}^0 L\left(p-\frac\varepsilon{2}\right)dp+
\varepsilon\int_0^pl'(p)dp=
  \int\limits_{-\frac\varepsilon{2}}^{\frac\varepsilon{2}} L(p)dp+\varepsilon(l(p)-l(0)).
\end{gather*}
As $L_0$ is odd and $l_0(0)=0$, this implies the needed.  
 
Finally, let us check that $L_0$ is continuous  up to the boundary of $\C_0$. 
As $(Q_0(p)+p)(Q_0(p)-p)=-1$,  the factor $\rho_0$ is continuous
in $\C_0$ and along its boundary. This and~\eqref{eq:R} imply that, 
being analytic in $\C_0$, $R_0$ is continuous up to its boundary.  
\end{proof}
\subsection{Completing the construction of the solution to the
Schr\"odinger equation}\label{ss:sol-construction}
Let $R_0$ be the solution to~\eqref{eq:R} described in 
Proposition~\ref{pro:R-func}. For $p\in\R$, we define
\begin{equation}
  \label{eq:R-for-Sch}
  R(p)=\begin{cases} R_0(p-i0),& p\le 0,\\ R_0(p+i0), & p>0.\end{cases}
\end{equation}
The function $R$ is continuous and satisfies~\eqref{eq:R} (we assumed that
$0<\varepsilon<1$). Let $p_1$ and $T$ 
be defined by~\eqref{K:formula}  and~\eqref{eq:T} with $Q=Q_0$.
One has 
\begin{Th} The constructed function $\Psi$  is continuous in $(t,x,p)\in 
(-\infty,1/\varepsilon]\times [0,\infty)\times \R$. Both for 
$0<x<1-\varepsilon t$ and for  $x>1-\varepsilon t$, it is 
infinitely differentiable in $x$ and $t$ and satisfies the 
Schr\"odinger equation~\eqref{non-stat-eq}.  At $x=0$, \  $\Psi$ satisfies 
the Dirichlet boundary condition.  At $x=1-\varepsilon t$, 
it  is continuously differentiable  in $x$.  
As a function of $p$, $\Psi$ is $\varepsilon$-periodic.
\end{Th}
\begin{proof} 
We need only to check that both the series in~\eqref{sol-form}
converge sufficiently well. By~(\ref{eq:R}), for $N\in\N$, 
\begin{equation}\label{R-and-series}
R(p+N\varepsilon)=\prod_{l=1}^{N} \rho(p+\varepsilon (l-1/2)) \  R(p).
\end{equation}
In view of Proposition~\ref{pro:R-func} and~\eqref{eq:R-for-Sch}, \ 
$R$ is bounded on the interval $[-1,1]$. This,~\eqref{R-and-series} 
and~\eqref{est:rho} imply that  $R(p+i0)=O(p^{-\infty})$ 
as $p\to+\infty$. As $R_0$ is even,  $R_0(p-i0)=O(|p|^{-\infty})$ 
as $p\to-\infty$. These estimates imply the needed.
\end{proof}
\section{Asymptotics of the function $L_0$ as $\varepsilon\to0$}
\label{sec:as-R}
The results of this section are used to  get the asymptotics of the 
solutions $\Psi_n$.
 
Below, $C$ denotes positive constants independent of $\varepsilon$. 
For any given function $z\to f(z)$, \ $O(f(z))$ denotes a function such that 
$|O(f(z))|\le C |f(z)|$. 
\subsection{The asymptotics ``between'' the branching points of $Q$}
For $a>0$, we define $K_0(a)=\{p\in\C\,:\, |\im p|>a(|\re p|-1)\}$.
One has 
\begin{Th}\label{th:L0inK0}
 Fix $a>0$. For $p\in K_0(a)$, 
\begin{equation}\label{as:L-in-K} 
  L_0(p)=l_0(p)+
O\left(\frac{\varepsilon^2(|p|+1)}{(p^2-1)^{3/2}}\right).
\end{equation}
\end{Th}
This theorem describes the asymptotic behavior of  $L_0$ 
both for $\varepsilon\to 0$ and for $p\to\infty$.
\begin{proof} 
Below, we consider only $p\in K_0(a)$.  Remind that
$\gamma(p)$, the integration contour in~(\ref{for:R}), is vertical.
Assume that $\gamma(p)\subset K_0(a)$.

Note that 
$l_0''(p)=O\left(\frac{|p|+1}{(p^2-1)^{3/2}}\right)$. This implies that, 
for $t\in\gamma(p)$, 
\begin{equation}
  \label{eq:l-on-gamma}
  l_0(t)=l_0(p)+l_0'(p)(t-p)+O(|t-p|^2M(t,p)),\ \ 
M(t,p)=\max_{\zeta\in\gamma_{p,t}(p)}\left|\frac{|\zeta|+1}{(\zeta^2-1)^{3/2}}\right|, 
\end{equation}
where $\gamma_{p,t}(p)$ is the segment of $\gamma(p)$ between $p$ and $t$.
Substituting this representation in the formula for $L_0$ 
in~(\ref{for:R}), we get
\begin{equation}\label{eq:L-on-gamma}
L_0(p)=l_0(p)+I(p),\quad I(p)=O\left(\frac1\varepsilon 
\int_{\gamma(p)}\frac{M(t,p) |t-p|^2 d\,\im t}
{\left|\cos^2\frac\pi\varepsilon(t-p)\right|}\right). 
\end{equation}
First, we consider the case where $\im p\ge 1$. Then, we can 
assume that the distance between $\gamma(p)$ and the points $\pm 1$
is bounded away from zero by a constant independent of $p$.  
We obtain 
\begin{equation*}
\begin{split}
  |I(p)|&\le \frac{C}\varepsilon\int\limits_{\gamma(p),\, \im t\le \frac{\im p}2}
\frac{|t-p|^2 d\,\im t}
{\left|\cos^2\frac\pi\varepsilon(t-p)\right|}+
\frac{C}{\varepsilon|p|^2}\int\limits_{\gamma(p), \, \im t\ge \frac{\im p}2}
\frac{|t-p|^2 d\,\im t}
{\left|\cos^2\frac\pi\varepsilon(t-p)\right|}\\
&\le \frac{C}\varepsilon\int\limits_{\gamma(p), \, \im t\le \frac{\im p}2}
e^{-\frac{2\pi}\varepsilon \im (p-t)} |t-p|^2 d\,\im t+
\frac{C\varepsilon^2}{|p|^2}\le \frac{C\varepsilon^2}{|p|^2}.
\end{split}
\end{equation*}
Now, consider the case of  $0\le \im p\le 1$.
We can choose the contour $\gamma(p)$ so that, for $t\in\gamma(p)$, \ 
$M(t,p)\le  C\,|p^2-1|^{-3/2}$. Then, 
\begin{equation*}
  |I(p)|\le \frac{C}{\varepsilon|p^2-1|^{3/2}}\int_{\gamma(p)}\frac{ |t-p|^2 d\,\im p}
{\left|\cos^2\frac\pi\varepsilon(t-p)\right|}=\frac{C\varepsilon^2}{|p^2-1|^{3/2}}.
\end{equation*}
The estimates for $I(p)$ 
imply~\eqref{as:L-in-K} in the case where $\im p\ge 0$. The complementary
case is treated similarly. 
\end{proof}
Replacing~(\ref{eq:l-on-gamma}) with a similar formula containing  more terms 
of the Taylor series for $l_0$, one proves
\begin{Le} \label{L-series}
Fix $a,\,\delta>0$ and $J\in\N$. In $K_0(a)$ without the 
$\delta$-neighborhood of $p=\pm 1$,
\begin{equation}
  \label{eq:L0-as}
  L_0(p)=\sum_{j=0}^{J-1} \varepsilon^{2j} l_j(p)+
O\left(\varepsilon^{2J}(1+|p|)^{-2J}\right),
\end{equation}
$l_j$ being analytic and satisfying the estimate  
$l_j(p)=O((1+|p|)^{-2j})$.
\end{Le}
\subsection{Analytic properties of $l_0$}\label{sec:l0}
In the sequel, we use some properties of $l_0$, the leading term
in the asymptotics of $L_0$. Let us discuss them here.  
 
Formula~(\ref{eq:l-def}) implies that
\point  for  $p\in\C_0$,
\begin{equation}
  \label{eq:l0near1}
  l_0(p)=\pi-2^{3/2}z+O(z^{3}),\quad z\to 0,
\end{equation} 
where $z=\sqrt{1-p}$, the function $p\mapsto z$  is analytic, 
and  $z>0$ if  $p<1$; 
\point $l_0$ is analytic in $z=\sqrt{1-p}$ in a neighborhood of $z=0$;
\point in $\C_+\cup(\R+i0)$, one has
\begin{equation}
  \label{eq:l0:as}
  l_0(p)=2i\ln(2p)+\pi+O(1/|p|^2),\quad |p|\to\infty,
\end{equation}
where the function $p\mapsto\ln p$ is analytic and defined so that $\ln i=i\pi/2$.
 
One also has
\begin{Le}\label{le:conf-prop}
The function $l_0$ conformally maps the first quadrant 
${\mathcal Q}_1=\{p\in\C:\, \im p,\,\re p>0\}$ onto the half-strip
$\Pi=\{z\in \C\,:\, 0<\re z<\pi,\,\im z>0\}$. The boundary of 
${\mathcal Q}_1$  is bijectively mapped onto the boundary of $\Pi$ 
in the following way:  $i\R_+$ is mapped onto itself, the interval 
$[0,1]$ is mapped onto the interval $[0,\pi]$, and the half-line 
$[1,+\infty)$ is mapped onto $\pi+i\R_+$.
\end{Le}
The statements of the lemma easily following from~(\ref{eq:l-def}), 
we omit the details. 

Finally, discuss the integral $\int_0^p l_0(p)\,dp$. 
Integrating by parts, we check that 
\begin{equation}\label{eq:int:l0}
\int_0^p l_0(p)\,dp=pl_0(p)-2iQ_0(p)-2,\quad p\in\C_0,
\end{equation}
where the integral is taken along a smooth curve in $\C_0$.
This formula implies, in particular, that
\begin{equation}\label{eq:int:l0:1}
\int_0^1 l_0(p)\,dp=\pi-2,
\end{equation}
and that
\begin{equation}
  \label{eq:int:l0-pi}
\int_1^p (l_0(p)-\pi)\,dp=2ip (\ln(2p)-1)+O(1/|p|),\quad
p\in \C_+\cup (\R+i0),\quad |p|\to\infty.
\end{equation}
\subsection{Asymptotic behavior of  $L_0$ near the point $p=1$}
The function $L_0$ is analytic in $\C_0$ and, in particular, between the points
$p=\pm1$.  We have assumed that $0<\varepsilon<1$. This allows 
to use~(\ref{eq:L}) to study $L_0$ to the right of $p=1$. The function $l_0'$ 
(staying in the right hand side in~(\ref{eq:L})) has a square root singularity 
at $p=1$: representation~(\ref{eq:l0near1}) shows that, in a neighborhood of 
$p=1$, $l_0'(p)=(2/(1-p))^{1/2}+O((1-p)^{1/2})$. Therefore, $L_0$ has  square root 
singularities  at $p=1+\varepsilon (1/2+l)$, $l=0,1,2,3\dots$. So, the asymptotic 
behavior of $L_0$ near the point $p=1$ is quite non-trivial. Here, we discuss
$L_0$ to left of $p=1$. And in the next section, we discuss its behavior to the 
right of $p=1$. 
\begin{Th} \label{th:as:L0}
Fix $a>0$ and sufficiently small $\delta>0$. 
In $K_0(a)$, in the $\delta$-neighborhood of $p=1$,
\begin{equation}
  \label{eq:L0near1}
  L_0(p)=\pi+\sqrt{2\varepsilon}\, \zeta\,(\,(p-1)/\varepsilon\,)+
O(\varepsilon^{3/2}+|p-1|^{3/2}),
\end{equation} 
where $\zeta$ is analytic  in $\C\setminus[1/2,\infty)$ and
defined there by the formula
\begin{equation}
  \label{eq:zeta}
  \zeta(t)=\lim_{L\to+\infty}
\left(\sum_{l=0}^{L-1}\frac1{\sqrt{l+1/2-t}}-2\sqrt{L}\right),
\end{equation}
the branch of the square root being fixed so that $\sqrt{\R_+}=\R_+$.
In  $C_a=\{t\in\C\,:\, |\im t|\ge a \re t\}$, the function $\zeta$ admits 
the asymptotic representation
\begin{equation}
  \label{eq:zeta-as}
  \zeta(t)=-2\sqrt{-t}+O(t^{-\frac32}),\quad |t|\to\infty.
\end{equation}
\end{Th}

\begin{Rem}
In the case of Theorem~\ref{th:as:L0},
the second term in~(\ref{eq:L0near1}) satisfies the estimate
\begin{equation}
  \label{eq:zeta-est}
  \sqrt{2\varepsilon}\,\zeta\,((p-1)/\varepsilon)=
O(\varepsilon^{1/2}+|p-1|^{1/2}).
\end{equation}
\end{Rem}
\begin{proof} Assume that $\delta$ is sufficiently small, and 
that $p\in K_0(a)$ is in the $\delta$-neighbor\-hood of $1$. 
In view of~\eqref{eq:l0:as}, one can use~(\ref{eq:l0near1}) 
as a rough approximation for $l_0$ on the whole curve $\gamma(p)$. 
Substituting~(\ref{eq:l0near1}) into the second formula 
in~(\ref{for:R}), we get
\begin{gather}\nonumber
  L_0(p)=\pi+\sqrt{2\varepsilon}\,\zeta((p-1)/\varepsilon)+\Delta,\\
\label{def:zeta:tmp}
  \zeta(t)=i\pi\,\int_{\gamma(t)}\frac{\sqrt{-s}\,ds}{\cos^2(\pi(t-s))},\qquad
\Delta=\frac1\varepsilon\int_{\gamma(p)}\frac{O((1-q)^{3/2})\,dq}
{\cos^2(\pi(p-q)/\varepsilon)}\,,
\end{gather}
where $\gamma(t)\subset C_a$, $\gamma(p)\subset K_0(a)$, and $\sqrt{-s}>0$ when 
$s<0$.
Using the Residue theorem, we obtain
\begin{equation*}
  \zeta(t)=\lim_{L\to\infty}\left(\sum_{l=0}^{L-1}\frac1{\sqrt{l+1/2-t}}-I_L\right),
\quad I_L=-i\pi\int_{\gamma(t-L)}\frac{\sqrt{-s}\,ds}{\cos^2(\pi(t-s))}.
\end{equation*}
In this formula, for $L$ sufficiently large, we can choose $\gamma(t-L)=i\R+t-L$. 
So, as $L\to\infty$,
\begin{align*}
   I_L&=-i\pi\int_{i\R}\frac{\sqrt{L-t-\tau}\,d\tau}{\cos^2(\pi\tau)}=\\
&=-i\pi\int_{-iL}^{iL}\frac{(\sqrt{L}+O(\tau/\sqrt{L}))\,d\tau}{\cos^2(\pi\tau)}+
O\,\left(\int_{|\im \tau|\ge L}\frac{\sqrt{|\tau|}\,d\im \tau}
{{\rm ch}^2\,(\pi\im\tau)}\right)=\\
&=2\sqrt{L}+O(1/\sqrt{L}).
\end{align*}
This implies~(\ref{eq:zeta}).
 
Let us estimate $\Delta$. Set $t=p-1$. First, we check that 
\begin{equation}
\label{est:Delta}
|\Delta|\le \frac{C}\varepsilon\int_{-\infty}^\infty e^{-2\pi|\im(t-\tau)|/\varepsilon}
(|t|+|\im \tau|)^{3/2}\,d\,\im\tau.
\end{equation}
One has $\Delta=\frac1\varepsilon
\int_{\gamma(t)}\frac{O(\tau^{3/2})\,d\tau}{\cos^2(\pi(t-\tau)/\varepsilon)}$.
If $\re t<0$, we choose $\gamma(t)=\{\tau\in\C\,:\,\re \tau=\re t\}$,
then, $|\tau|\le |\re t|+|\im \tau|\le |t|+|\im\tau|$. This leads 
to~\eqref{est:Delta}. If $\re t\ge 0$, we pick 
$\gamma(t)=\{\tau\in\C\,:\,\im\tau=k(t)\,\re \tau\}$ with $k(t)=\im t/\re t$. 
Clearly, $|k(t)|\ge a$. Now, we have  $|\tau|\le 
|\re \tau|+|\im \tau|\le (1+1/a)|\im\tau|$. This again implies~\eqref{est:Delta}. 

If $|t|\le\varepsilon$, then, after the change of the variable 
$\tau:=\tau/\varepsilon$ in~\eqref{est:Delta}, we see that  
$\Delta=O(\varepsilon^{3/2})$. Otherwise, we use the estimates
\begin{align*}
|\Delta|&\le \frac{C}\varepsilon
\left(\int_{-2|t|}^{2|t|}+\int_{-\infty}^{-2|t|}+\int_{2|t|}^{\infty}\right) 
e^{-2\pi|\im t-y|/\varepsilon}
(|t|+|y|)^{3/2}\,d\,y\\
&\le\frac{C}\varepsilon \int_{-\infty}^{\infty} 
e^{-2\pi|\im t-y|/\varepsilon}|t|^{3/2}dy+\frac{C}\varepsilon
\left(\int_{-\infty}^{-2|t|}+\int_{2|t|}^{\infty}\right)
e^{-2\pi |\im t-y|/\varepsilon}|y|^{3/2}\,dy.
\end{align*}
They imply that $|\Delta|\le C|t|^{3/2}$. The obtained estimates 
for $\Delta$ justify the error term estimate in~(\ref{eq:L0near1}).

Finally, let us study the function $\zeta$. Its analyticity 
follows from~\eqref{def:zeta:tmp}. For $t\in C_a$, 
sufficiently large $|t|$ and $l\ge 0$,
\begin{equation*}
\frac12\,(l+1/2-t)^{-1/2}=(l+1-t)^{1/2}-(l-t)^{1/2}+O((l-t)^{-5/2})
\end{equation*}
uniformly in $t$ and $l$. Substituting this representation into~(\ref{eq:zeta}), 
we easily arrive at~\eqref{eq:zeta-as}. This competes the proof.  
\end{proof}
\subsection{The asymptotics of $L_0$ to the right of $p=1$} 
\label{sec:L1}
Fix $a>0$. Here, we describe the asymptotics for $L_0$ in the domain 
$K_1(a)=\{p\in\C_+\,:\,\im p\ge a\,\re (1-p)\}$.

Let $\C_1=\C\setminus\{z=x\in\R\,:\,x\le1\}$. We continue analytically 
the function $l_0$ from $\C_0$ to $\C_1$ across  $\C_+$ and denote 
the obtained function by $l_1$.

Let  $L_1$ be a solution to the equation 
\begin{equation}
  \label{eq:L1}
L_1(p+\varepsilon/2)-L_1(p-\varepsilon/2)=\varepsilon l_1'(p), \quad 
p\pm\varepsilon/2\in \C_1.
\end{equation}
Assume that  $L_1$ is analytic in $\C_1$. Then, in $\C_+$, both the functions 
$L_0$ and $L_1$ are analytic and  satisfy one and the same difference equation. 
This implies that,  in $\C_+$,  the function $P=L_0-L_1$ is an analytic and  
$\varepsilon$- periodic. To get the asymptotics of $L_0$, we
construct $L_1$ and to analyze   $L_1$ and $P$. 

The function $L_1$ is constructed similarly to $L_0$.
It equals the right hand side 
of the second formula from~(\ref{for:R}) 
with $\gamma(p)\subset \C_1$. One has
\begin{equation}\label{as:L1-in-D} 
  L_1(p)=l_1(p)+O\left(\frac{\varepsilon^2p}{(p^2-1)^{3/2}}\right),
\quad p\in K_1(a), 
\end{equation}
This representation is proved in the same way as~(\ref{as:L-in-K}).

Fix $\delta>0$. Reasoning as in the proof of Theorem~\ref{th:as:L0}, 
one checks that,  in the $\delta$-neighborhood of $p=1$ in $K_1(a)$, 
\begin{equation}
  \label{eq:L1near1}
  L_1(p)=\pi-i\sqrt{2\varepsilon}\, \zeta\,(\,(1-p)/\varepsilon\,)+
O(\varepsilon^{3/2}+|p-1|^{3/2}).
\end{equation} 

Finally, we discuss the function $P$. Its Fourier series is described by
\begin{Le} \label{le:P} For $p\in \C_+$,  
\begin{gather}
  \label{for:P}
P(p)=\sum_{k=1}^\infty e^{2\pi i k (p-1-\varepsilon/2)/\varepsilon}P_k,\\
\label{for:Pk}
P_k=2e^{i\pi /4} \left(\frac{\varepsilon}{k}\right)^{1/2}+
O\left( \left(\frac{\varepsilon}{k}\right)^{3/2}\right).
\end{gather}
\end{Le}
\begin{proof}
Note that  
\begin{equation*}
 l_1(1)=l_0(1)=\pi,\quad\text{and}\quad 
l_1(1-it+0)-l_1(1) =-(l_0(1-it-0)-l_0(1)),\quad t>0.
\end{equation*}
This and the integral representations for $L_0$ and $L_1$ imply that
\begin{equation*}
 P(p)=\frac{\pi}{i\varepsilon}\int_{1-i\infty}^{1}\frac{(l_0(\zeta)-l_0(1))\,d\zeta} 
{\dsize \cos^2 \frac {\pi(p-\zeta)}\varepsilon},\quad \im p>0. 
\end{equation*}
Integrating by parts, we get
\begin{equation*}
  P(p)=-i\int\limits_{1-i\infty}^{1} 
\left(\tan\frac{\pi(p-\zeta)}\varepsilon-i\right)\,l_0'(\zeta)\,d\zeta
=\int\limits_{1-i\infty}^{1} \frac{-4i\,d\zeta}
{ \left(e^{-2\pi i(p-\zeta)/\varepsilon}+1\right)\,Q_0(\zeta)},
\end{equation*}
where $Q_0$ is the branch of the function $\zeta\mapsto\sqrt{\zeta^2-1}$ 
from the definition of $l_0$. As $p\in\C_+$, this implies~\eqref{for:P} with
%
$P_k=4i\int\limits_{-i\infty}^{0} \frac{e^{-2\pi i k t/\varepsilon}\,dt}{Q_0(t+1)}$.
%
Clearly,
\begin{equation*}
P_k=4i\int\limits_{-i}^{0} e^{-2\pi i k t/\varepsilon}\,
(1+O(t))\,\frac{dt}{\sqrt{2t}}+\int_{-i\infty}^{-i}
e^{-2\pi i k t/\varepsilon}\,O(1/t)\,dt
\end{equation*}
with $\sqrt{t}=e^{3i\pi/4} |t|$ for $t\in i\R_-$.
This implies~\eqref{for:Pk}. 
\end{proof}
We finish this section with
\begin{Cor}\label{cor:L0inK1} For $p$ in the first quadrant 
$\mathcal Q_1$, 
  \begin{equation*}
i\int_{p_0}^pL_0(p)\,dp=-2p\ln|p| +O(|p|),\quad \quad |p|\to\infty.
  \end{equation*}
\end{Cor} 
\begin{proof} In $\C_+$, one has
$\int_{p_0}^pL_0\,dp=\int_{p_0}^pP\,dp+\int_{p_0}^pL_1\,dp$.  
In view of Lemma~\ref{le:P}, $\int_{p_0}^pP\,dp$ is bounded  in 
$\C_+$. So, representations~\eqref{as:L1-in-D} 
and~\eqref{eq:int:l0-pi} imply the needed.
\end{proof}
\section{Standard asymptotic behavior of  $\Psi_n$}
\label{sec:standard}
Here, we prove Theorem~\ref{th:norm-wave:up}, 
Lemma~\ref{le:large-xi-est} and outline
the proof of Theorem~\ref{th:norm-wave:down}.
We fix $n>0$ and use the notations introduced in 
Section~\ref{ss:notations}. 
\subsection{A convenient integral representation for $\Psi$}
\label{sec:wedge:saddle point}
First, we check                                             
\begin{Le}\label{le:psi-n:up} Pick $0<\theta<\pi/2$. One has
\begin{gather}
\label{eq:Psi-n:up-2}
\Psi_n(x,t)=\frac{e^{it}}{\sqrt{\varepsilon\pi}}
\int_{e^{i\theta}\R} A(p)\sin(px)
e^{\frac{i}\varepsilon S(p,\tau)}\,dp,\quad \tau=\varepsilon t, \\
\label{eq:S}
S(p,\tau)=p^2(1-\tau)-2\pi np+\int_0^{p} l_0(p)\,dp,\\
\label{eq:A}
A(p)=e^{\frac{i}\varepsilon\int_0^p (L_0(p)-l_0(p))\,dp}.
\end{gather}
\end{Le}
\begin{proof}
In view of formulas~\eqref{eq:R-for-Sch} and~\eqref{for:R},  
the integral in~\eqref{eq:Psi-n:up} equals the integral of 
$A(p)\sin(px)e^{\frac{i}\varepsilon S(p,\tau)}$ taken along the 
path going in $\C_0$, first, along $\R-i0$ from $-\infty$ to $-1$, next, 
along $\R$  to $1$, and then, along $\R+i0$ to $+\infty$. 
Thanks to Corollary~\ref{cor:L0inK1}, one can deform 
the part of the integration path going along $\R_+$ to $e^{i\theta}\R_+$.  
As $p\to\int_0^pL_0(p)\,dp$  is even, see~\eqref{eq:L-prop}, one can deform 
the whole contour to $e^{i\theta}\R$. 
\end{proof}
\begin{Rem}Fix  sufficiently small $a,b>0$. Let $V(b)$ be the 
$b$-neighborhood of the points $\pm 1$. Fix $M\in\N$.  In view 
of Lemma~\ref{L-series}, the factor $A$ admits the uniform  asymptotic 
representation 
\begin{equation}\label{as:A}
A(p,\varepsilon)=\sum_{m=0}^{M-1}\varepsilon^mA_m(p)+O(\varepsilon^{M}),\quad 
p\in K_0(a)\setminus V(b),\quad
\varepsilon\to 0,
\end{equation} 
with $A_0\equiv1$ and $A_l$ independent of $\varepsilon$, analytic and bounded 
uniformly in  $p$.
\end{Rem}
In the next sections, we study the asymptotic behaviour of the integral 
in~(\ref{eq:Psi-n:up-2}) as $\varepsilon\to0$ by means of the method of steepest 
descents, see, e.g.~\cite{Wong}.
\subsection{Saddle point and lines of steepest descent}
\label{sec:saddle-and-steepest:1}
Discuss the {\it saddle points}, i.e., the zeros of the function 
$p\mapsto S_p(p,\tau)$. One has 
\begin{Le}\label{le:saddle-point}
Fix $\tau<\tau_n$, $\tau_n$ being defined in~\eqref{eq:tau_N}. 
Then, in $\C_0$, there is only one zero of $S_p$. Denote it by $p_n(\tau)$.
The $p_n$ is simple, is located on $[0,1]$ and satisfies~\eqref{eq:pn-up}. 
One has $S_{pp}(p,\tau)>0$ if $-1<p<1$.
\end{Le}
\begin{proof}
Differentiating~\eqref{eq:S}, we get 
\begin{equation}
  \label{eq:S-prime}
S_p(p,\tau)=2p(1-\tau)+l_0(p)-2\pi n.
\end{equation}
Lemma~\ref{le:conf-prop} and~\eqref{eq:l0-sym} imply that, in $\C_0$, \ 
$\im S_p(p,\tau)$ vanishes only on $[-1,1]$. As $\tau<\tau_n<1$, 
Lemma~\ref{le:conf-prop} and formula~\eqref{eq:S-prime} imply 
that $S_{p}(.,\tau)$ is monotonously 
increasing on $[-1,1]$.  As $S_p(0,\tau)= 
l_0(0)-2\pi n=-2\pi n<0$ and as $S_p(1,\tau)=2(1-\tau)+\pi -2\pi n=
2(\tau_n-\tau)>0$, we conclude that, in $\C_0$, \ $S_p$ has a unique 
zero $p_n$, that it is simple and that $0<p_n<1$. Formulas~\eqref{eq:S-prime} 
and~\eqref{eq:l0-arsin} imply~\eqref{eq:pn-up}. Finally, using the inequality
$\tau<1$ and~\eqref{eq:l0-arsin},  one checks that
$S_{pp}(p,\tau)> 0$ for $-1<p<1$. 
\end{proof}
Now, we discuss  the paths of steepest descents ``beginning'' at $p_n(\tau)$.
Remind that they are described by the equation 
$\re S(p,\tau)=\re S(p_0,\tau)$. As $S_{pp}(p_n,\tau)>0$, there are four 
of them, and, at $p=p_n$, the angles between $\R_+$ and these curves  
are equal to $\pi/4+\pi l/2$, $l=0,1,2,3$. We denote the paths of steepest 
decent by $\gamma_l$, $l=0,1,2,3$, respectively. Along $\gamma_0$ and 
$\gamma_2$, $\im S(p,\tau)$ is monotonously increasing 
as $p$ moves away from $p_n$.  One has  
\begin{Le}\label{le:steep-des-cur} If $\tau<\tau_n$, then
$\gamma_0,\gamma_2\subset\C_0$, $\gamma_0$ goes to infinity inside $\C_+$, and 
$\gamma_2$ goes  to infinity inside $\C_-$ . Along these curves,
\begin{equation}\label{eq:as-st-des-line}
\im p=\re p+O(\ln\re p), \quad |p|\to \infty.
\end{equation} 
\end{Le}
\begin{proof}
We prove the statement concerning  $\gamma_0$. 
The analysis of $\gamma_2$ is similar.\\     
{\bf 1.}  Prove that $\gamma_0$ goes to infinity inside $\C_+$.
Consider the  values of $S$ on $\R+i0$. 
Pick a real  $a\ne p_n$. Using~\eqref{eq:S-prime}, the inequality $\tau<1$ 
and Lemma~\ref{le:conf-prop}, we check that $(\re S)_p$ monotonously increases 
along $\R+i0$. Therefore, $\re (S(a,\tau)-S(p_n,\tau))\ne 0$. So,  
$\gamma_0$ can not connect $p_n$ to the point $a$ inside $\C_+$.
Furthermore, the maximum principle for the harmonic functions implies that
 $\gamma$ can not come back to $p_n$.
These observations imply the needed.\\ 
{\bf 2.} \ Let us study $\gamma_0$ for large $|p|$.   
Formula~\eqref{eq:int:l0-pi} implies that
\begin{equation}\label{as:intl0}
  S(p,\tau)=p^2(1-\tau)+O(p\ln(2p)),\quad p\in \C_+\cap(\R+i0),\quad 
|p|\to\infty.
\end{equation}
As, along $\gamma_0$,  $\re S(p,\tau)=\re S(p_n,\tau)$ and 
$\im S(p,\tau)$ increases as $p\to\infty$,~\eqref{as:intl0}
leads to~\eqref{eq:as-st-des-line}.
\end{proof}
\subsection{Proof of Theorem~\ref{th:norm-wave:up}}
\label{sec:proof:th:norm-wave:up}
Let, as before, $V(b)$ be the $b$-neighborhood of the points $\pm 1$.
Lemma~\ref{le:steep-des-cur} implies that there exist sufficiently 
small positive numbers $a$ and $b$  such that  
$\gamma_{0},\,\gamma_2\,\subset K_0(a)\setminus V(b)$.  The paths of steepest 
decent continuously depend on the parameter $\tau$. So, there are 
$a, b>0$ such that $\gamma_{0},\,\gamma_2\,\subset K_0(a)\setminus V(b)$
for all $T_1<\tau<T_2$, i.e., for all $\tau=\varepsilon t$ 
considered in Theorem~\ref{th:norm-wave:up}. 

We deform the integration path in integral in~(\ref{eq:Psi-n:up-2}) 
to $\gamma=\gamma_0\cup\gamma_2$ and replace the factor  $A$ in the 
integrand by the expression in the right hand side in~\eqref{as:A}. 
As a result, the integral in the right hand side in~(\ref{eq:Psi-n:up-2}) 
becomes the sum of  integrals along $\gamma$. Applying the method of steepest 
decent to these integrals,  we arrive to the asymptotic
representation
\begin{gather*}\label{eq:saddle-point:1}
  \Psi_n(x,t)=\sqrt{\frac{2}{S_{pp}(p_n,\tau)}}\,
e^{\frac{i(\tau+ S(p_n,\tau))}\varepsilon+\frac{i\pi}4} 
\left(\sin (p_nx)+\sum_{k=1}^{L-1}\varepsilon^k\psi_{n,k}(x,\tau)+
O(\varepsilon^{L})\right),\\
\tau=\varepsilon t,\qquad \varepsilon\to 0,
\end{gather*}
where the coefficients  $\psi_{n,k}$  are bounded uniformly in 
$0\le x\le 1-\tau$ and $T_1<\tau<T_2$, and  
the error estimate is uniform in $0\le x\le 1-\varepsilon t$ and 
$T_1<\varepsilon t<T_2$. One competes the proof of Theorem~\ref{th:norm-wave:up} 
using
\begin{Le}\label{le:S-at-pn} For $\tau<\tau_n$,
\begin{equation}\label{eq:legeandre:up}
 \tau+S(p_n(\tau),\tau)=\int_{\tau}^{\tau_n}E_n(\tau)\,d\tau+2\tau_n-3,\quad
\frac1{S_{pp}(p_n(\tau),\tau)}=\frac12\frac{d\ln p_n}{d\tau}(\tau), 
\end{equation}
where $E_n(\tau)=p_n^2(\tau)-1$. 
\end{Le}
\begin{proof} 
As $S_{p}(p_n,\tau)=0$, we get
\begin{equation*}
  \frac{dS(p_n(\tau),\tau)}{d\tau}
=S_p(p_n(\tau),\tau)\,p_n'(\tau)+
S_\tau(p_n(\tau),\tau)=S_\tau(p_n(\tau),\tau)
=-p_n^2(\tau).
\end{equation*}
Therefore,
\begin{equation}\label{eq:S-en-pn-1}
\tau+S(p_n(\tau),\tau)=\tau_n+S(p_n(\tau_n),\tau_n)+
\int_\tau^{\tau_n}E_n(\tau)\,d\tau.  
\end{equation}
Note that $p_n(\tau_n)=1$. This and~\eqref{eq:int:l0:1} imply that
\begin{equation}\label{eq:S-en-pn-2}
  \tau_n+S(p_n(\tau_n),\tau_n)=1-2\pi n+\int_0^{1} l_0(p)\,dp=
-1-2\pi(n-1/2)=2\tau_n-3.
\end{equation}
Formulas~\eqref{eq:S-en-pn-1} and~\eqref{eq:S-en-pn-2} imply
the first relation in~\eqref{eq:legeandre:up}. 
Using the definition of $p_n$, we get 
\begin{equation*}
  0=\frac{d}{d\tau}S_p(p_n(\tau),\tau)=
S_{pp}(p_n(\tau),\tau)\,p_n'(\tau)+S_{\tau p}(p_n(\tau),\tau).
\end{equation*}
This leads to  the second relation in~\eqref{eq:legeandre:up}.
\end{proof}
\subsection{Outside the potential well}\label{ss:outside}
\subsubsection{}{\it The proof of  Theorem~\ref{th:norm-wave:down}} 
is parallel to the proof of Theorem~\ref{th:norm-wave:up}. We only 
outline it. Below,   $\xi=\varepsilon (x-(1-\tau))$, 
$\tau=\varepsilon t$,  $\xi\ge 0$ and  $\tau<1$.

\vspace{2mm}\noindent
Instead of~\eqref{eq:Psi-n:up-2}, we obtain 
\begin{gather}
\label{eq:psi-n:down-2}
\Psi_n(x,t)=\frac{(-1)^{n+1}e^{it-i\xi/2-i\varepsilon(1-\tau)/4}}{\sqrt{\varepsilon\pi}}\,
\operatornamewithlimits\int_{e^{i\theta}\R} \tilde A(p)\, p\,
e^{\frac{i}\varepsilon \tilde S(p,\tau,\xi)}\,dp,
\\
\label{eq:tildeS,tildeA}
\tilde S(p,\tau,\xi)=S(p,\tau)+Q_0(p)\xi,\qquad
\tilde A(p)=A(p)\,e^{
-\frac{i}\varepsilon\,\int_{p-\varepsilon/2}^p(L_0(q)-l_0(p))\,dq}.
\end{gather}
We apply the  method of steepest descents
to the integral in~\eqref{eq:psi-n:down-2}. One has
\begin{Le} If $\xi>0$, then in $\C_0$ there is  only one saddle point 
$\tilde p_n(\tau,\xi)$ of the function $p\mapsto \tilde S(p,\tau,\xi)$. 
It is simple, satisfies~\eqref{eq:pn:down} and is located 
in the first quadrant $\QV_1$. If $\tau<\tau_n$ and $\xi=0$, then  
$\tilde p_n(\tau,\xi)=p_n(\tau)$. 
\end{Le}
\noindent
The paths of steepest decent  continuously depend  on $\xi$. This allows to use 
for them the notations introduced  in the case when  $\xi=0$ and $\tau<\tau_n$.
With these notations, the statement of Lemma~\ref{le:steep-des-cur} remains true.
 
\vspace{2mm}\noindent 
Applying the method of steepest decent and using the formulas
\begin{gather*}
 \tau+\tilde S(\tilde p_n(\tau,\xi),\tau,\xi)=\int_{\tau}^{\tau_n}E_n(\tau)\,d\tau
+\int_0^\xi Q_0(\tilde p_n(\tau,\xi'))\,d\xi'+(2\tau_n-3),\\
\frac1{\tilde S_{pp}(\tilde p_n(\tau,\xi),\tau,\xi)}=
\frac12\frac{\partial\ln \tilde p_n(\tau,\xi)}{\partial\tau},
\end{gather*}
analogous to formulas~\eqref{eq:legeandre:up}, one obtains
representation~(\ref{as:Psi:down:1}).
\subsubsection{Proof of Lemma~\ref{le:large-xi-est}}\label{sss:Lemma}
If $\xi$ is sufficiently large, the integration path in~\eqref{eq:psi-n:down-2} 
can be deformed to $i\R$. Indeed,  Theorem~\ref{th:L0inK0} implies that, for 
any fixed $a>0$, in $K_0(a)$,  the factor $\tilde A$ stays bounded  as 
$p\to\infty$.  On the other hand, in view  of~\eqref{eq:int:l0-pi}, there 
is a constant $C>0$ such that $\im S(p,\tau)\ge -C|p|$ as $p$ tends to infinity 
inside the sector $0\le \arg p\le \pi/2$. As $l_0$ is odd, one has an estimate of 
the same form  in the sector $-\pi\le \arg p\le -\pi/2$. These observations 
imply the needed.
   
As along the imaginary axis $l_0(p)\in i\R$, see~\eqref{eq:l0-sym}, for 
sufficiently large $\xi$,  we get
\begin{equation*}
  |\Psi_n(x,t)|\le \frac{1}{\sqrt{\varepsilon\pi}}\,\sup_{p\in i\R}|\tilde A|\,
\int_{-\infty}^{\infty}|t|e^{\frac1\varepsilon(2\pi n t-\sqrt{1+t^2}\xi)}dt.
\end{equation*}
To complete the proof of the lemma, we estimate   
the last integral by means of the Laplace method, see~\cite{Wong}. 
We omit elementary details. 
\section{Destruction of the standard  adiabatic behavior}
\label{sec:st-behavior-destruction}
Here, we prove Theorem~\ref{th:tau-smaller-tau-n}.  
When $\tau=\varepsilon t$ increases,   $\tau< \tau_n$, 
the saddle point $p_n$ in~(\ref{eq:Psi-n:up-2}) moves to the point $p=1$, 
a branch point of the action $S$ (a branch point of 
$l_0$ in~\eqref{eq:S}). After a natural change of variables,
$S$ becomes an analytic function having two saddle points approaching 
one to another as $\tau\to\tau_n$. This effect determines the asymptotic
behavior  of $\Psi_n$ for $\tau\sim\tau_n$. 
\subsection{The factor $A$}
To control the factor $A$ in~(\ref{eq:Psi-n:up-2}),
we often use quite a rough 
\begin{Le}\label{le:A-rough} 
For sufficiently small $\varepsilon$,
  \begin{equation}
    \label{est:A-rough}
 \frac1\varepsilon\int_0^p(L_0(p)-l_0(p))\,dp=O(\varepsilon^{1/2}),\quad p\in\C_0.
  \end{equation}
\end{Le}
\begin{proof} 
In view of~(\ref{eq:l0-sym}) and~(\ref{eq:L-prop}), it suffices to 
prove~(\ref{est:A-rough}) only for $p\in {\mathcal Q}_1$. Assume that
$p\in {\mathcal Q}_1$ and fix $a>0$. If $p\in K_0(a)$ and $|p-1|\ge \varepsilon$, 
the statement follows from Theorem~\ref{th:L0inK0}. If $p\in K_0(a)$ and 
$|p-1|\le \varepsilon$, we prove~(\ref{est:A-rough}) using 
also~(\ref{eq:L0near1}) and~(\ref{eq:l0near1}). 

Assume that $p\in K_1(a)$. Then, the analysis is based on the formula 
$L_0=L_1+P$ and representations~(\ref{for:P})~--~(\ref{for:Pk}) describing 
the function $P$. If  $p\in K_1(a)$ and $|p-1|\ge\varepsilon$, we come 
to~(\ref{est:A-rough}) using also~(\ref{as:L1-in-D}) to control $L_1$, and if 
$p\in K_1(a)$ and $|p-1|\le\varepsilon$, then, in 
addition, we use~(\ref{eq:L1near1}).   
\end{proof}
\subsection{Reducing the problem to a ``local one''}\label{ss:loc}
Pick  $0<b<1$. Denote by $V(b)$ the $b$-neighborhood of $p=1$. 
Let  $\tau_n-\delta\le \tau\le \tau_n$ for a $\delta>0$.
Bellow, we assume that $\delta$ is sufficiently small. 
Then, in particular, the saddle point $p_n(\tau)$ is in $V(b)$.
 
Consider $\gamma_0$ and $\gamma_2$, two paths of the steepest 
descents beginning at $p_n$. Let $\gamma=\gamma_0\cup\gamma_2$, and let 
$\gamma(b)$ be the connected component of $\gamma\cap V(b)$ containing 
$p_n$. There is $C>0$ such that, as $\varepsilon\to 0$,
\begin{equation}
\label{eq:Psi-n:up-3}
\Psi_n(x,t)=\frac{e^{i\frac\tau\varepsilon}}{\sqrt{\varepsilon\pi}}
\int_{\gamma(b)} A(p)\sin(px) e^{\frac{i}\varepsilon S(p,\tau)}\,dp
+O(e^{-C/\varepsilon}).
\end{equation}
One proves this representation using an argument standard for the method 
of steepest descents. Let us outline it. First, fix $\tau$. Note that the 
derivative of  $\im  S(p,\tau)$  along $\gamma$ equals $|S_p(p,\tau)|$, 
and that, along $\gamma$, $S_p'(p,\tau)\to \infty$ as $p\to\infty$, 
and $S_p'(p,\tau)\ne 0$ if $p\ne p_n(\tau)$. 
Using these observations, one checks that 
$\int_{\gamma\setminus\gamma(b)} A\,\sin(px) 
e^{\frac{i}\varepsilon S}\,dp=O(e^{-C/\varepsilon})$.
As $S$, $p_n$ and the paths of steepest decent continuously depend on 
$\tau$, this estimate is uniform both in $x$ and $\tau$.  
\subsection{Local analysis of the action $S$}
Let us study  $S$ near the point $p=1$. We begin with

\begin{Le}\label{le:S-in-q} For $p\in\C_0$ and $|p-1|<1$, consider  the action 
$S$ as a function of the variable $z=z(p)=\sqrt{1-p}$ fixed by the condition 
$z(p)>0$ for $p<1$. One has
\begin{gather}\label{eq:S-in-q}
  S(p(z),\tau)=
S(1,\tau)-2(\tau_n-\tau)\,z^2+\frac{4\sqrt{2}}{3}\,z^3+(1-\tau)z^4+z^5f(z),
\\
\label{eq:Sat1}
  S(1,\tau)=-3+2\tau_n-\tau,
\end{gather}
where $f$ is independent of $\tau$ and
analytic in $z$ in the $1$-neighborhood of zero.
\end{Le}
\begin{proof} 
Using formulas~(\ref{eq:S}),~(\ref{eq:tau_N}), we obtain
\begin{equation}\label{eq:S:p-1}
S(p,\tau)- S(1,\tau)=
2(\tau_n-\tau)\,(p-1)+(1-\tau)\,(p-1)^2
+\int_1^p(l_0(p)-\pi)\,dp.
\end{equation}
This and~(\ref{eq:l0near1}) imply the representation
\begin{equation}\label{eq:S:p-1-2}
S(p,\tau)- S(1,\tau)=2(\tau-\tau_n)\,(1-p)+(1-\tau)\,(1-p)^2+
\frac{4\sqrt{2}}{3}(1-p)^{\frac32}+O((1-p)^{\frac52}).
\end{equation}
This implies~(\ref{eq:S-in-q}). The analyticity of $\tilde f$ follows  
from the analyticity of $l_0$ in $z$. Finally, by means of~(\ref{eq:int:l0:1}), 
we get $S(1,\tau)=-1+\pi-2\pi n-\tau=-3+2\tau_n-\tau$. 
\end{proof}
Denote by $S_z$ and $S_{zz}$ the first and second derivatives 
of the function $z\to S(p(z),\tau)$ with respect to $z$.
By means of Lemma~\ref{le:S-in-q}, one checks
\begin{Le}\label{le:S-in-q:1} 
If $b>0$ is sufficiently small, then there exists $\delta=\delta(b)$
such that for all $|\tau-\tau_n|\le \delta$, \
$S_z$ has only two zeros in $z$ in the  $\sqrt{b}$-neighborhood $z=0$. 
They are simple if $\tau\ne \tau_n$, coincide 
if $\tau=\tau_n$, and are  located at the points
$$z=0\quad\text{and}\quad z=z_n(\tau):=z(p_n(\tau)).$$
The function $\tau\to z_n(\tau)$ is analytic in $\tau$, and
\begin{equation}\label{eq:zn-as}
z_n(\tau)=\frac1{\sqrt{2}}(\tau_n-\tau)+O((\tau-\tau_n)^2),\quad 
\tau\to\tau_n.
\end{equation}
\end{Le} 
Lemmas~\ref{le:S-in-q} and~\ref{le:S-in-q:1} show that after the change of 
variable $p\to z(p)$ the asymptotic analysis of $\Psi_n$ is reduced  to 
the standard asymptotic analysis of an integral with two coalescing saddle 
points, a well understood classical problem see, e.g.,~\cite{Wong}.
\subsection{Change of variables}
The key to the asymptotic analysis of integrals with two nearby saddle points
is a change of variables that transforms the action into a third order polynomial. 
This transformation is described by the Chester, Friedmann and Ursell  
theorem~\cite{Ch-Fr-Ur,Wong}. We formulate it in a convenient for us form. 
For $z$ close to $0$, we set
\begin{equation}
  \label{eq:G-S}
G(z,\tau)=\frac{S(p(z),\tau)-S(1,\tau)}2=-(\tau_n-\tau)\,z^2+
\frac{2\sqrt{2}}{3}\,z^3+O(z^4),
\end{equation}
and, for  $\tau$ close to $\tau_n$, using the formula
\begin{equation}\label{eq:lambda}
\lambda(\tau)=\left(-6G(z_n(\tau),\tau)\right)^{1/3},
\end{equation}
we define a function  $\lambda$ 
analytic in $\tau$ and such that $\lambda(\tau)>0$ when $\tau<\tau_n$ and 
\begin{equation}\label{eq:lambda-as}
\lambda(\tau)=\tau_n-\tau+o(\tau_n-\tau),\qquad \tau\to\tau_n.
\end{equation}
Define a mulivalued analytic function $z\to u(z,\tau)$ by the equation
\begin{equation}\label{eq:G-u}
G(z,\tau)=-\lambda(\tau)u^2+\frac{2\sqrt{2}}{3}\,u^3.
\end{equation}
One has
\begin{Th} \label{th:CFU}
If $c>0$ is sufficiently small, then, for sufficiently small $\delta(c)>0$,
there is just one branch of the function $u$ that is analytic in 
$(z,\tau)\in\{|z|<c\}\times\{|\tau-\tau_n|<\delta(c)\}$. For this branch,
\begin{equation}\label{saddle-points-connection}
u(0,\tau)=0\quad\text{and}\quad 
u(z_n(\tau),\tau)=\lambda(\tau)/\sqrt{2}.
\end{equation}
The correspondence $z\leftrightarrow u$ is $(1,1)$.
\end{Th} 
For $|\tau-\tau_n|<\delta(c)$, we define $u\mapsto z(u,\tau)$, the 
function inverse to $z\to u(z,\tau)$. Let $u_n(\tau)=u(z_n(\tau),\tau)=
\lambda(\tau)/\sqrt{2}$. 
We shall use
\begin{Cor} For $|\tau-\tau_n|<\delta(c)$, 
  \begin{equation}\label{eq:derivatives}
  \left.\frac{z}{u}\,  \frac{\partial z}{\partial u}\right|_{u=0}= 
\frac{\lambda(\tau)}{\tau_n-\tau}\,,\qquad
  \left.\frac{z}{u}\,\frac{\partial z}{\partial u}\right|_{u=u_n(\tau)}=
{\sqrt{\frac1{\lambda(\tau)}\,\frac{d\ln p_n}{d\tau}(\tau)}}\,,
  \end{equation}
where the square root positive for $\tau<\tau_n$.
\end{Cor}

\begin{proof}
As $z(0,\tau)=0$, we get 
$\left.\frac{z}{u}\,\frac{\partial z}{\partial u}\right|_{u=0}=
\left.\left(\frac{\partial z}{\partial u}\right)^2\right|_{u=0}$.
Furthermore, relations~(\ref{eq:G-u}) 
and~(\ref{eq:G-S}) imply that
$\left.\left(\frac{\partial z}{\partial u}\right)^2\right|_{u=0}=
\lambda(\tau)/(\tau_n-\tau)$. This proves the
first formula in~(\ref{eq:derivatives}). 
Prove the second one. Using~(\ref{eq:G-u}), we obtain 
$\frac12\left(\left(\frac{\partial z}{\partial u}\right)^2S_{zz}+ 
\frac{\partial^2 z}{\partial u^2} S_z\right)=2(-\lambda+2\sqrt{2}u).$
Therefore, $\left.\left(\frac{\partial z}{\partial u}\right)^2\right|_{u=u_n}
\left.S_{zz}\right|_{z=z_n}=4(-\lambda+2\sqrt{2}u_n)=4\sqrt{2} u_n$.
On the other hand, as $p=1-z^2$, \ $S_{zz}=4z^2S_{pp}-2S_p$. This implies 
that $\left.S_{zz}\right|_{z=z_n}=4z_n^2\left.S_{pp}\right|_{p=p_n}$. Using these 
two observations and the second formula in~(\ref{eq:legeandre:up}), we obtain
$$\left.\frac{z_n^2}{u_n^2}\;
\left(\frac{\partial z}{\partial u}\right)^2\right|_{u=u_n} 
=\frac{z_n^2}{u_n}\;\frac{4\sqrt{2}}{\left.S_{zz}\right|_{z=z_n}}=
\frac{\sqrt{2}}{u_n\left.S_{pp}\right|_{p=p_n}}=
\frac{2}{\lambda\left.S_{pp}\right|_{p=p_n}}=
\frac1{\lambda}\,\frac{d\ln p_n}{d\tau}(\tau).$$
This implies the second formula in~(\ref{eq:derivatives})
up to the sign of its right hand side.  The first formula 
in~\eqref{eq:derivatives} and~\eqref{eq:lambda-as} imply that 
$\left.\frac{z}{u}\,\frac{\partial z}{\partial u}\right|_{u=0,\,\tau=0}=1$.
This implies that the sign in the second formula in~(\ref{eq:derivatives})
is correct and  completes the proof.
\end{proof}
\subsection{Integration path}
Let $\beta(b,\tau)=u(z(\gamma(b)),\tau)$, $\gamma(b)$ being the curve defined
in section~\ref{ss:loc}, and let 
\begin{equation}\label{mathcal-S}
{\mathcal S}(u,\lambda)=
-\lambda u^2+\frac{2\sqrt{2}}3\,u^3.
\end{equation} 
Relation~\eqref{eq:G-u} imply
\begin{Le}\label{le:B-path} For the function $u\to {\mathcal S}(u,\lambda(\tau))$,
\point $u_n(\tau)$ is a saddle point, and $\beta(b,\tau)$ is a 
segment of $B(\tau)$, the path of steepest descent containing $u_n(\tau)$ 
and such that $\im {\mathcal S}(u,\lambda(\tau))$ increases as $u$ moves 
away from  $u_n(\tau)$ along $B(\tau)$;
\point there is a positive constant $C(b)$ such that, for all 
$\tau$ we consider, $\im {\mathcal S}(u,\lambda(\tau))>C(b)$ at the ends of 
$\beta(b,\tau)$;
\point the orientation of $\beta$ induced by the orientation of $\gamma(b)$ 
is such that, at $u=u_n$, the $\beta(b,\tau)$ is oriented downwards;
\point $B(\tau)$ is a smooth curve having the asymptotes $-i\R_+$ 
and $e^{\frac{i\pi}6}\R_+$.
\end{Le}
\begin{proof} The first two properties follow from the analogous 
properties of $p_n$ and $\gamma(b)$. To prove the third one,  
remind that, at $p_n$, $\gamma(b)$ is oriented upwards.
In view of the second formula in~(\ref{eq:derivatives}), $u_z|_{z_n}>0$. As
$\left.\frac{\partial u}{\partial p}\right|_{p_n}=
-\left.\frac1{2z_n}\,\frac{\partial u}{\partial z}\right|_{z_n}<0$,
at the point $u=u_n$, the path $\beta(b,\tau)$ is oriented downwards.
The proof of the forth  property is elementary and is omitted.
\end{proof}
\subsection{The proof of Theorem~\ref{th:tau-smaller-tau-n}}
Define $F$ by~\eqref{eq:F}. One has
\begin{Pro} \label{le:leading-term-tau-less-than-tau-n}
As $\varepsilon\to 0$, 
  \begin{gather}\label{eq:I_b-as-tau-smaller-tau-n}
    \Psi_n(x,t)=a_0(x,\tau)F(e^{i\pi/6}Z_n)+
O(\varepsilon^{2/3}(1+|Z_n|^{\frac12})),\quad
Z_n=\lambda(\tau)/(4\varepsilon)^{1/3},\\
\label{eq:a0-temp}
    a_0(x,\tau)=\left(4\varepsilon\right)^{1/6}
e^{\frac{i}\varepsilon (S(1,\tau)+\tau)+\frac{i\pi}4}\,
    \sin(p_n(\tau)x)\,\,\left.\frac{z}{u}\,
\frac{\partial z}{\partial u}\right|_{u_n(\tau)},\quad \tau=\varepsilon t. 
  \end{gather}
\end{Pro}
\begin{proof}
Let $b>0$ and $\delta>0$ be sufficiently small, and let $\tau_n-\delta\le\tau
\le \tau_n$. In the integral in~\eqref{eq:Psi-n:up-3} we change the variable 
$p\mapsto u=u(z(p),\tau)$. The leading term in~\eqref{eq:Psi-n:up-3} takes 
the form 
\begin{equation}
   \label{eq:Ib}
   \Psi_n^{(0)}(x,\tau)=
-\frac{2e^{\frac{i}\varepsilon (S(1,\tau)+\tau)}}{\sqrt{\pi\varepsilon}} 
\int_{\beta(b,\tau)} A(p)
\varphi(u,\tau,x) e^{\frac{2i}\varepsilon{\mathcal S}(u,\lambda(\tau)) }\,u\,du,
\end{equation}
where
\begin{equation}
  \label{eq:varphi}
\varphi(u,\tau,x)=
\sin(px)\,\frac{z}{u}\frac{\partial z}{\partial u},\quad
p=1-z^2,\quad z=z(u,\tau).
\end{equation}
Represent $\varphi$  in the form
\begin{equation}
  \label{eq:varphi01}
  \varphi(u,\tau,x)=\varphi_0(\tau,x)+(u-u_n(\tau))\varphi_1(u,\tau,x),
\quad \varphi_0(\tau,x)=\varphi(u_n(\tau),\tau,x).
\end{equation}
%
%
Note that ${\mathcal S}_u'(u,\tau)=2\sqrt{2}u(u-u_n(\tau))$.
Using integration by parts, we  get 
\begin{gather*}
\Psi_n^{(0)}(x,\tau)=k\,(J_1+J_2+J_3+J_4),\quad 
k=-\frac{2e^{\frac{i}\varepsilon (S(1,\tau)+\tau)}}{\sqrt{\pi\varepsilon}},\\
J_4= \int_{\beta(b,\tau)} (A(p)-1)\varphi(u,\tau,x) 
e^{\frac{2i}\varepsilon{\mathcal S} }\,u\,du,\qquad
J_1= \varphi_0(\tau,x)\int_{\beta}
e^{\frac{2i}\varepsilon{\mathcal S} }\,u\,du,\\
J_2=\frac{\varepsilon}{4\sqrt{2}i}\left.\varphi_1(u,\tau,x)\, 
e^{\frac{2i}\varepsilon{\mathcal S}}
\right|_{\beta}, \qquad
J_3=\frac{i\varepsilon}{4\sqrt{2}}\int_{\beta}(\varphi_1)_u'(u,\tau,x)\,
e^{\frac{2i}\varepsilon{\mathcal S}}\,du,
\end{gather*}
where ${\mathcal S}={\mathcal S}(u,\lambda)$, \ $\lambda=\lambda(\tau)$.
In view of the second statement of Lemma~\ref{le:B-path}, 
the term $J_2$ is exponentially small in $\varepsilon$, 
and, up to an exponentially small error, we can replace 
in the formula for $J_1$ the integration path by  $B(\tau)$.
We get
\begin{equation}\label{eq:J1J2}
J_1+J_2=\varphi_0(\tau,x)\int_{B(\tau)}
e^{\frac{2i}\varepsilon {\mathcal S}(u,\lambda) }\,u\,du+O(e^{-C/\varepsilon}).
\end{equation}
In the integral in~(\ref{eq:J1J2}), we change the variable
$u\mapsto v=e^{-\frac{5\pi i}6} 2^{\frac56}\varepsilon^{-\frac13}(u-2^{-\frac32}\lambda)$. 
This gives
\begin{gather}
  \label{eq:J1J2:1}
J_1+J_2=-2^{-\frac23}\pi e^{\frac{i\pi}6} \,\varepsilon^{\frac23}\varphi_0(\tau,x)
e^{-\frac{2i}3 Z_n^3}f(e^{i\pi/6}Z_n)+O(e^{-C/\varepsilon}),\\
\label{eq:F(z)}
f(s)=\frac1{2\pi i}\int_{\gamma} e^{-\frac{v^3}3+s^2v}\,(s-v)\,dv,
\end{gather}
where $Z_n$ is as in~\eqref{eq:I_b-as-tau-smaller-tau-n},
and $\gamma$ is a smooth curve  going from \ $e^{-2i\pi/3}\,\infty$ \ 
to  \ $e^{2i\pi/3}\,\infty$. Recall that 
$\Ai(s)=\frac1{2\pi i}\int_{\gamma} e^{-\frac{v^3}3+s v}\,dv$, see~\cite{Olver}.
Therefore, $f(s)=s\Ai(s^2)-\Ai'(s^2)$. This implies that the 
leading term in~\eqref{eq:J1J2:1} being multiplied by the factor 
$k$ is equal to the leading term in~\eqref{eq:I_b-as-tau-smaller-tau-n}. 
To complete the proof, it suffices to check the estimates 
\begin{equation}\label{eq:J3}
J_3= O(\varepsilon^{4/3}/(1+|Z_n|^{1/2})),\quad 
J_4=O(\varepsilon^{7/6}(1+|Z_n|^{1/2})).
\end{equation} 
Begin with $J_3$.  If $|\lambda/\varepsilon^{1/3}|\le 1$,
we change the variable $u$ to $w=u/\varepsilon^{1/3}$ and get 
$J_3=\varepsilon^{4/3}\int_{\beta/\varepsilon^{1/3}}O(1)
e^{2i{\mathcal S}(w,\lambda/\varepsilon^{1/3})}\,dw=
O(\varepsilon^{4/3})$. If $|\lambda/\varepsilon^{1/3}|\ge 1$, 
we change the variable $u$ to  $w=u/\lambda$. Then,  
$J_3=\varepsilon\lambda \int_{\beta/\lambda}O(1)
e^{\frac{2i\lambda^3}\varepsilon{\mathcal S}(w,1)}\,dw$,
the integral is taken along a segment of a path of the steepest descent of 
the function $w\mapsto {\mathcal S}(w,1)$, and this segment contains 
the saddle point $w=1/\sqrt{2}$. Direct 
application of the method of steepest descents leads to the 
estimate $J_3=O(\varepsilon^{3/2}/\lambda^{1/2})=
O(\varepsilon^{4/3}/|Z_n|^{1/2})$. Our estimates for $J_3$ 
imply the first estimate  in~(\ref{eq:J3}).
Turn to  $J_4$. Using Lemma~\ref{le:A-rough}, we get  
$J_4=\int_{\beta}O(\varepsilon^{1/2})
e^{\frac{2i}\varepsilon {\mathcal S}(u,\lambda)}u\,du$ and
estimate $J_4$ reasoning as when estimating $J_3$.
\end{proof}
Let us complete the proof of Theorem~\ref{th:tau-smaller-tau-n}. 
First, check that $Z_n(\varepsilon t)$ admits the representation 
from this theorem. Using~(\ref{eq:Sat1}) and the first 
relation in~(\ref{eq:legeandre:up}), we check that  $S(p_n(\tau),\tau)-
S(1,\tau)=\int_{\tau}^{\tau_n}E_n(\tau)\,d\tau$. This 
and~\eqref{eq:lambda} imply the representation $\lambda(\tau)=
\left(3\int_{\tau_n}^{\tau}E_n(\tau)\,d\tau\right)^{\frac13}$.
Substituting it into the formula for $Z_n$ 
from~\eqref{eq:I_b-as-tau-smaller-tau-n}, we get the  needed 
representation for $Z_n$. Next, we transform the expression for $a_0$ 
from~(\ref{eq:I_b-as-tau-smaller-tau-n}).
Using~(\ref{eq:Sat1}), the second relation 
in~(\ref{eq:derivatives}),~\eqref{eq:psi-n} and the expression for 
$Z_n$ from~\eqref{eq:I_b-as-tau-smaller-tau-n}, we prove that
$a_0=c_n\,\sqrt{(\ln p_n)'(\tau)/Z_n} \,\,\, \psi_{n}(x,\tau)$.
This formula  and~(\ref{eq:I_b-as-tau-smaller-tau-n}), 
imply~(\ref{psi:n:tau-less-tau-n}). \qed
\section{Aftermath: the case of $\varepsilon t\ge\tau_n$}
\label{sec:tau-greater-tau-n}
Here, first, we  single out the terms that become the leading
terms of the asymptotics of $\Psi_n$ for $t$ and $x$ 
satisfying~\eqref{eq:tau-ge-tau-n}. Then, we describe the asymptotic 
behavior of each of these terms and prove Theorem~\ref{th:final}. 
\subsection{Singling out the leading terms}
For $\tau\ge\tau_n$, the main contribution to the integral 
in~\eqref{le:psi-n:up} comes from a small neighborhood of the 
point $p=1$, the branch point of the action $S$. Let us formulate the precise
statement. Pick $0<\alpha<1/6$ and  set $\rho(\varepsilon)=
\varepsilon^{\frac23-\alpha}$. For $a,b\in \C$,  denote by $[a,b]$ 
the segment of straight line connecting $a$ and $b$. We prove
\begin{Pro}\label{pro:wedge:main-contrib}
For $t$ and $x$ satisfying~\eqref{eq:tau-ge-tau-n},
as $\varepsilon\to 0$,
\begin{equation}\label{eq:TFRE}
\Psi_n(x,\varepsilon t)=
e^{i\phi(\varepsilon)}({\mathcal T}(x,\varepsilon t)+
{\mathcal R}(x,\varepsilon t)+{\mathcal G}(x,\varepsilon t))+
E(x,\varepsilon t),\\
\end{equation}
where  
\begin{equation}\label{eq:T-def}
{\mathcal T}(x,\tau)=\frac{e^{it}}{\pi^{\frac12}\varepsilon^{\frac12}}
\left(\int_{1-i\rho(\varepsilon)}^1+\int_1^{1+\rho(\varepsilon)}\right)
\sin(px)\,e^{\frac{i}\varepsilon S(p,\tau)}\,dp,
\end{equation}
\begin{equation}
\label{eq:R-def}
{\mathcal R}(x,\tau)=
\frac{ie^{it}}{\pi^{\frac12}\varepsilon^{\frac32}}\,
\int_1^{1+\rho(\varepsilon)}\sin(px)\,
e^{\frac{i}\varepsilon S(p,\tau)}\,\int_1^pP(q)\,dq\,dp,
\end{equation}
\begin{equation}
\label{eq:F-def}
\begin{split}  
{\mathcal G}(x,\tau)=
\frac{ie^{it}}{\pi^{\frac12}\varepsilon^{\frac32}}&
\left(\int_{1-i\rho(\varepsilon)}^1\right.
\sin(px)\,e^{\frac{i}\varepsilon S(p,\tau)}\,
\int_1^p(L_0(q)-l_0(q))\,dq\,\,dp+\\
&\left.\int_{1}^{1+\rho(\varepsilon)}\sin(px)\,
e^{\frac{i}\varepsilon S(p,\tau)}\,
\int_1^p(L_1(q)-l_0(q))\,dq\,dp\right),
\end{split}
\end{equation}
$E(x,\tau)=O(\varepsilon^{\frac76})$, and 
$\phi(\varepsilon)=\frac1\varepsilon\,\int_0^1(L_0-l_0)\,dp=
O(\varepsilon^{\frac12})$.
\end{Pro}
\noindent
To prove the  proposition, we  use the following two lemmas. 
\begin{Le} \label{le:path-deform} For $p\in\C_0$, we let   
$p=1+e^{i\phi}s$, $s\ge0$, $\phi\in\R$. 
If  $0<\phi<\pi$, then
\begin{equation*}
\im S(p,\tau)=(1-\tau)s^2\sin(2\phi)+2s(\ln(2s)-1)\cos\phi
+2 s\,\sin\phi (\tau_n-\tau-\phi)+O(1)
\end{equation*}
as $s\to \infty$. If $-\pi<\phi<0$, then, as $s\to\infty$,
\begin{equation*}
\im S(p,\tau)=(1-\tau)s^2\sin(2\phi)-
2s(\ln(2s)-1)\cos\phi+ 2s\sin\phi\,(\tau_n-\tau+\phi)+
O(1).
\end{equation*} 
These representations are uniform in $\phi$.
\end{Le}
\begin{proof} The first statement follows from  
representations~(\ref{eq:S:p-1}) and~(\ref{eq:int:l0-pi}).
The second  follows from the first one and the relation 
$\overline{S(\bar p,\tau)}=S(p,\tau)$ valid for $p\in\C_0$.
\end{proof}
In the next lemma, we discuss $S(\cdot , \tau)$ on 
$\Gamma=\{1-i\R_+\}\cup\{1\}\cup\{1+\R_++i0\}\subset\C_0$. 
\begin{Le}\label{le:action} Let $\tau\ge \tau_n$.
\point Along $\Gamma$, \ $\im S(p,\tau)$ monotonously increases as $p$
moves away from $p=1$, and $\partial \im S(p,\tau)/\partial |p-1|>0$ when $p\ne1$;
\point Fix $b>0$. There is $C>0$ such that,
along $\Gamma$, for $|p-1|\le b$, one has  $\im S(p,\tau)>C|p-1|^{\frac32}$.
\end{Le}
\begin{proof}
The first statement immediately follows from the following properties of $S$:
\\ 1) If $p>1$, then $\im S(p+i0,\tau)=\int_1^{p}\im l_0(p+i0)\,dp$, and, 
according to Lemma~\ref{le:conf-prop}, $\im l_0(p+i0)>0$ for $p>1$. 
\\ 2) If  $p\in 1-i\R_+$, then 
$\im S=-2(\tau-\tau_n) \im(p-1)-\int_{\im p}^0\re(l_0(1+is)-\pi)\,ds$, 
and, in view of Lemma~\ref{le:conf-prop}, \ $\re l_0(p)<\pi$ inside $\C_0$.
\\ The second statement follows the first one and 
representation~(\ref{eq:S:p-1-2}).
\end{proof} 
Now, we turn to the proof of Proposition~\ref{pro:wedge:main-contrib}.
\begin{proof} 
Lemma~\ref{le:path-deform} and estimate~\eqref{est:A-rough} imply that, 
for sufficiently small $\varepsilon$ and $\tau\ge\tau_n$, one can deform 
the integration path in~(\ref{eq:Psi-n:up-2}) to $\Gamma$.
Let $\Gamma_+=\{1\}\cup(1+\R_++i0)$, and let $\Gamma_-=\{1\}\cup(1-i\R_+)$.
Lemma~\ref{le:path-deform} also shows that, as $p\to\infty$, 
\begin{gather}
\im S(p,\tau)=2p\ln p+O(p),\quad p\in\Gamma_+,\\
\im S(p,\tau)=2(\tau-\tau_n+\pi/2)\,|\im p|+O(1),\quad p\in\Gamma_-.
\end{gather}
Using this and the first point of Lemma~\ref{le:action}, one proves that,
for any fixed $b>0$,  there is a $C>0$, such that, for sufficiently small 
$\varepsilon$, modulo $O(e^{-C/\varepsilon})$,  \ $\Psi_n$ equals  the 
right hand side  of~(\ref{eq:Psi-n:up-2}) with the integration path replaced 
by $\Gamma\cap\{|p-1|\le b\}$. Finally, choosing $b$ sufficiently small 
and using the second point of Lemma~\ref{le:action}, we prove that, 
modulo $O(e^{-C/\varepsilon^{\frac{3\alpha}2}})$, 
$\Psi_n$ equals to the right hand side of~(\ref{eq:Psi-n:up-2}) 
with the integration path replaced by $\Gamma\cap\{|p-1|\le \rho(\varepsilon)\}$.

Now, discuss the factor $A$ in the integrand in~(\ref{eq:Psi-n:up-2}).
In view of~(\ref{est:A-rough}),
  \begin{equation}
    \label{eq:A+}
    A(p)=e^{i\phi(\varepsilon)}\left(1+
    \frac{i}\varepsilon\int_1^p(L_0-l_0)\,dp+O(\varepsilon)\right),
\quad p\in\Gamma.
  \end{equation}
We use this representation for $p\in\Gamma_-$. For $p\in\Gamma_+$, 
we represent $L_0$ in the form $L_0=L_1+P$, see Section~\ref{sec:L1}, and get 
  \begin{equation}
    \label{eq:A-}
    A(p)=e^{i\phi(\varepsilon)}\left(1+
    \frac{i}\varepsilon\int_1^p(L_1-l_0)\,dp+
\frac{i}\varepsilon\int_1^pP(p)\,dp+O(\varepsilon)\right),
\quad p\in\Gamma_+.
  \end{equation} 
Estimate~\eqref{est:A-rough} also 
implies that $\phi(\varepsilon)=O(\varepsilon^{1/2})$.
Substituting~\eqref{eq:A-} and~\eqref{eq:A+} in~\eqref{eq:Psi-n:up-2} 
with the integration path replaced with 
$\Gamma\cap\{|p-1|\le \rho(\varepsilon)\}$ and with the corresponding error
estimate, we arrive at~\eqref{eq:TFRE} with   
\begin{equation}
  \label{eq:E}
  E=\int_{\Gamma\cap\{|p-1|\le \rho(\varepsilon)\}}
O(\varepsilon^{\frac12})e^{\frac{i}\varepsilon S(p,\tau)}\,dp
+O(e^{-C/\varepsilon^{\frac{3\alpha}2}}).
\end{equation}
Using the second point of Lemma~\ref{le:action}, we get
$|E|\le C\,\varepsilon^{\frac12}\,\int_0^\infty 
e^{-\frac{C}{\varepsilon}|p-1|^{\frac32}}\,dp\le 
C\,\varepsilon^{\frac76}$. 
This completes the proof of the proposition.
\end{proof} 
\subsection{The term ${\mathcal T}$}
Here, we prove that, in the case of~\eqref{eq:tau-ge-tau-n},  
as $\varepsilon\to 0$, \ 

\begin{equation}\label{T-T_0}
{\mathcal T}(x,\tau)={\mathcal T}_0(x,\tau)+
O\left(\varepsilon/(1+|z_n(\tau)|^{3})\right)
\end{equation}
with  $z_n$ and ${\mathcal T}_0$ given by~\eqref{as:T}.  
\begin{proof}
{\bf 1.} \ In the integral in~\eqref{eq:T-def}, we pass to the variable 
$z=z(p)$ as in Lemma~\ref{le:S-in-q}. The integration path turns into
$c(\varepsilon^{1/3-\alpha/2})$, where, for $r>0$,
$$ c(r)=[-ir,0]\cup[0,e^{i\pi/4}r]\subset -i\R_+\cup\{0\}\cup e^{i\pi/4}\R_+.$$
As a curve, $c(r)$ is oriented downwards. Since $0<\alpha<1/6$,
on the integration path, the expression 
$z^4/\varepsilon=O(\varepsilon^{1/3-2\alpha})$ is small, 
and, in view of~\eqref{eq:S-in-q} and~\eqref{mathcal-S},
$$\sin(p(z)x)\,e^{\frac{i}\varepsilon S(p(z),\tau)}=
e^{\frac{i}\varepsilon (S(1,\tau)+2 {\mathcal S}(z,\tau_n-\tau))}(\sin x+O(z^2)
+O(z^4/\varepsilon))$$ 
Therefore,
\begin{gather}\label{T:I}  
{\mathcal T}(x,\tau)=-\frac{2c_ne^{-i\pi/4}}{(\pi\varepsilon)^{1/2}}
(\sin x\cdot I_0+I_1+I_2/\varepsilon),\\
\nonumber
I_0=\int_{c} e^{\frac{2i}\varepsilon {\mathcal S}}\,z\,dz, \quad
I_1=\int_{c} e^{\frac{2i}\varepsilon {\mathcal S}}\,O(z^3)\,dz,\quad
I_2=\int_{c}e^{\frac{2i}\varepsilon{\mathcal S}}\,O(z^5)\,dz,
\end{gather}
where ${\mathcal S}= {\mathcal S}(z,\tau_n-\tau)$ and 
$c=c(\varepsilon^{1/3-\alpha/2})$.
\\
{\bf 2.} Along $c(\infty)$, one has 
$\left|e^{\frac{2i}\varepsilon{\mathcal  S}(z,\tau_n-\tau)}\right|\le e^{-C|z|^3/\varepsilon}$.
Therefore,
$I_0=\operatornamewithlimits\int_{c(\infty)}
e^{\frac{2i}\varepsilon{\mathcal  S}}\,z\,dz
+O(e^{-C\varepsilon^{-3\alpha/2}})$.
The last integral is computed as the analogous  integral from the proof of 
Proposition~\ref{le:leading-term-tau-less-than-tau-n}. This gives
\begin{equation}\label{eq:T:I0}
  I_0=-2^{-2/3}\sqrt{\pi} e^{i\pi/4} \varepsilon^{2/3} F(e^{i\pi/6}z_n(\tau))
+O(e^{-C\varepsilon^{-3\alpha/2}}).
\end{equation}
{\bf 3.} \ To estimate the integral $I_1$, we 
change the variable $z$ to $w=z\varepsilon^{-1/3}$ and obtain
$I_1=\varepsilon^{4/3} 
\operatornamewithlimits
\int_{c(\varepsilon^{-\alpha/2})}e^{2i{\mathcal S}(w,4^{1/3}z_n(\tau))} O(w^3)\,dw$.
Therefore,  if $-1\le z_n(\tau)\le 0$, then 
$I_1=O(\varepsilon^{4/3})$.
Let $z_n(\tau)\le -1$. Changing the variable $w$ to $u=-w/z_n(\tau)$,  we 
get
$ I_1=\varepsilon^{4/3} z_n^4(\tau)\int_{c(b)}
e^{-2iz_n^3{\mathcal S}(u,-4^{1/3})}O(u^3)\,du$ with
$b=-\varepsilon^{-\alpha/2}/z_n(\tau)$, 
and integrating twice by parts, we obtain 
$I_1= O(\varepsilon^{4/3} /z_n^{2})$.
\\
{\bf 4.} Estimating the integral $I_2$ similarly, one checks that
$I_2=O(\varepsilon^2)$  if $|z_n|\le 1$, and that $I_2=O(\varepsilon^2/|z_n|^{3})$
otherwise.
\\
{\bf 5.} \ 
Representation~\eqref{T:I}, formula~(\ref{eq:T:I0}) for $I_0$ and 
our estimates for $I_1$ and $I_2$ lead to~\eqref{T-T_0}.
\end{proof}
\subsection{The term ${\mathcal R}$}
Here, we prove that, under conditions~\eqref{eq:tau-ge-tau-n}, 
\begin{equation}\label{R-R_0}
{\mathcal R}(x,\tau)={\mathcal R}_0(x,\tau)+O(\varepsilon^{4/3}),
\quad \varepsilon\to 0,
\end{equation}
with ${\mathcal R}_0$ given by~\eqref {as:R}.
\begin{proof} 
{\bf 1.} \ Put
\begin{equation}
  \label{eq:Q:def}
  Q(p)=\frac1{\varepsilon^{\frac32}}\,\int_1^pP(q)\,dq.
\end{equation}
Representation~(\ref{for:P}) implies that, 
for $p\in\R$, as $\varepsilon\to 0$,
\begin{equation}
  \label{eq:Q:as}
  Q(p)=
\frac{e^{\frac{-i\pi}4}}{\pi}f\left(\frac{p-1}\varepsilon\right)+O(\varepsilon),\quad
f(p)=\sum_{k=0}^\infty f_k e^{2\pi i k p},
\end{equation}
where $f_k$ are given in~(\ref{eq:fk}).\\ 
{\bf 2.} \ Using the second point of 
Lemma~\ref{le:action},~\eqref{eq:Q:as} and  the representation
$\sin(px)=\sin x+O(p-1)$, one proves that
\begin{equation}
  \label{eq:R:1}
  {\mathcal R}(x,\tau)=
\frac{e^{it+\frac{i\pi}4}\sin x}{\pi^{\frac32}}\,\int_{1}^{1+\rho(\varepsilon)}
e^{\frac{i}\varepsilon S(p,\tau)}\,
f\left({\dsize\frac{p-1}\varepsilon}\right)\,dp+
O(\varepsilon^{\frac43}).
\end{equation} 
{\bf 3.} \ Using~(\ref{eq:S:p-1-2}), we rewrite~\eqref{eq:R:1} in the form
\begin{equation*}
  {\mathcal R}(x,\tau)=
\frac{c_n\sin x}{\pi^{\frac32}}\,\int_0^{\rho}
e^{\frac{i}\varepsilon \left(2(\tau_n-\tau)\,t+i\frac{4\sqrt{2}}{3}t^{\frac32}+
(1-\tau)\,t^2+O(t^{\frac52})\right)}\,f(t/\varepsilon)\,dt+O(\varepsilon^{\frac43}).
\end{equation*} 
{\bf 4.} \ As, in the definition of $\rho$, $0<\alpha<1/6$, 
the last integral can be transformed to the form 
\begin{equation*}
\int_0^{\rho}
e^{\frac{i}\varepsilon 
\left(2(\tau_n-\tau)\,t+i\frac{4\sqrt{2}}{3}t^{\frac32}\right)}\,\left(1+
i(1-\tau)\,t^2/\varepsilon+O(t^{\frac52}/\varepsilon)+
O(t^4/\varepsilon^2)\right)\,f(t/\varepsilon)\,dt.
\end{equation*} 
As, for $c>0$ and $l>0$, both the integrals
$\int_0^{\infty} e^{-\frac{c}{\varepsilon}t^{\frac32}}
t^{\frac52}/\varepsilon\,dt$ and
$\int_0^{\infty} e^{-\frac{c}{\varepsilon}t^{\frac32}}\,
t^{4}/\varepsilon^2\,dt$ are $O(\varepsilon^{\frac43})$,
and 
$\int_\rho^{\infty} e^{-\frac{c}{\varepsilon}t^{\frac32}}
t^l\,dt=O(e^{-C\varepsilon^{-3\alpha/2}})$,
we get 
\begin{equation*}
  \label{eq:R:3}
  {\mathcal R}(x,\tau)=
\frac{c_n\sin x}{\pi^{\frac32}}\,\int_0^{\infty}
e^{\frac{2i}\varepsilon \left((\tau_n-\tau)\,t+
i\frac{2\sqrt{2}}{3}t^{\frac32}\right)}\,(1+
i(1-\tau)\,t^2/\varepsilon)\,f(t/\varepsilon)\,dt+O(\varepsilon^{\frac43}).
\end{equation*} 
{\bf 5.} \ Consider the integral
$I_1=\int_0^{\infty}
e^{\frac{2i}\varepsilon \left((\tau_n-\tau)\,t+
i\frac{2\sqrt{2}}{3}t^{\frac32}\right)}\,f(t/\varepsilon)\,dt$.
Representing the function $f$ by its Fourier series and using the definition of 
$\tau_N$, see~(\ref{eq:tau_N}), we get
%
$ I_1=\sum_{k=0}^\infty f_k\int_0^{\infty}
e^{\frac{2i}\varepsilon \left((\tau_{n-k}-\tau)\,t+
i\frac{2\sqrt{2}}{3}t^{\frac32}\right)}\,dt$.
Changing  the variable $t\mapsto u=2^{\frac56}\varepsilon^{-\frac13}t^{\frac12}$,
and using the definition of the function $a$, see~(\ref{eq:a(lambda)}),
we get finally
\begin{equation*}
  I_1=\left(\frac\varepsilon2\right)^{\frac23}\,\sum_{k=0}^\infty f_k\,
a\left(\frac{\tau_{n-k}-\tau}{(4\varepsilon)^{\frac13}}\right).
\end{equation*}
{\bf 6.} \ Similarly, one proves that
\begin{equation*}
I_2:=\frac1\varepsilon\int_0^{\infty}
e^{\frac{2i}\varepsilon \left((\tau_n-\tau)\,t+
i\frac{2\sqrt{2}}{3}t^{\frac32}\right)}\,f(t/\varepsilon)\,t^2\,dt=
-\frac\varepsilon{16}\,\sum_{k=0}^\infty f_k\,
a''\left(\frac{\tau_{n-k}-\tau}{(4\varepsilon)^{\frac13}}\right). 
\end{equation*}
{\bf 7.} \ The last three steps lead to the formula~\eqref{R-R_0}.
\end{proof}
\subsection{The term ${\mathcal G}$}

Here, we assume that the condition~\eqref{eq:tau-ge-tau-n} is satisfied, and
we check that, as $\varepsilon\to 0$,  
\begin{equation}
  \label{as:G_G0}
\begin{split}
  {\mathcal G}(x,\tau)&=O(\varepsilon^{\frac23}),\quad \text{and}\\
  {\mathcal G}(x,\tau)&={\mathcal G}_0(x,\tau)+
O\left(\varepsilon^{2/3}/z^{5/2}\right) \quad \text{if} \quad 
z=(\tau-\tau_n)/\varepsilon^{1/3}\to\infty,
\end{split}
\end{equation}
${\mathcal G}_0$ being described in~\eqref{as:F:tau-ge-tau-n:2}.
\begin{proof} 
According to~(\ref{eq:F-def}), we can write
${\mathcal G}={\mathcal G}_-+{\mathcal G}_+$, where 
\begin{equation}\label{eq:Fpm}
{\mathcal G}_-=\int_{1-i\rho(\varepsilon)}^1g_-(p,\tau) 
e^{\frac{iS(p,\tau)}\varepsilon}\,dp,\quad
{\mathcal G}_+=\int_1^{1+\rho(\varepsilon)}g_+(p,\tau) 
e^{\frac{iS(p,\tau)}\varepsilon}\,dp,
\end{equation}
$g_-$ and $g_+$ are functions analytic in 
$\C_0=\C\setminus\{p\in\R:\,|p|\ge 1\}$ and 
$\C_1=\C\setminus\{p\in\R:\,p\le 1\}$  respectively and
satisfy the estimates 
\begin{equation}\label{eq:fpm}
  g_-=O(1), \ p\in [1,1-i\rho(\varepsilon)],\quad g_+=O(1), \ 
1\le \re p\le 2,\ -1\le \im p\le 0,
\end{equation}
that follow from estimate~(\ref{est:A-rough}) for the integral
of $L_0-l_0$ and its analogue for the integral of $L_1-l_0$.

The estimate ${\mathcal G}=O(\varepsilon^{\frac23})$ follows directly 
from~\eqref{eq:Fpm},~\eqref{eq:fpm} and 
the second point of Lemma~\ref{le:action}.  
The key to the representation in the second line of~\eqref{as:G_G0}
is 
\begin{Le} Pick $0<\beta<1$. Set $\rho_0=
\frac{\varepsilon}{\tau-\tau_n}\,z^{\beta}$. If $\varepsilon\to0$ and 
$z\to\infty$, then
\begin{equation}
  \label{eq:Gpm:1}
{\mathcal G}_\pm=\mp e^{\frac{iS(1,\tau)}\varepsilon}
\int_{1-i\rho_0}^1g_\pm e^{-\frac{2i(\tau-\tau_n)(p-1)}\varepsilon}\,dp+
O\left(\varepsilon^{\frac23}z^{-\frac52}\right).
\end{equation}
\end{Le}
\begin{proof}
Let us explain the way to obtain~\eqref{eq:Gpm:1} for ${\mathcal G}_+$. 
For the  analytic continuation of $S$ from $C_+$ to $\C_1$ across $1+\R_+$, 
we keep the ``old'' notation $S$.  It suffices to justify the equalities
\begin{align*}
{\mathcal G}_+(x,\tau)
&=\int_1^{1+e^{-i\pi/4}\rho} g_+ e^{\frac{iS(p,\tau)}\varepsilon}\,dp+
O\left(e^{-C\varepsilon^{-3\alpha/2}}\right)\\
&=\int_1^{1+e^{-i\pi/4}\rho_0} g_+ e^{\frac{iS(p,\tau)}\varepsilon}\,dp+
O\left(\varepsilon^{\frac23}e^{-Cz^{\beta}}+e^{-C\varepsilon^{-3\alpha/2}}\right)\\
&=\int_1^{1+e^{-i\pi/4}\rho_0} g_+
e^{\frac{i}\varepsilon\left(S(1,\tau)-2(\tau-\tau_n)(p-1)\right)} \,dp+
O\left(\varepsilon^{\frac23}z^{-\frac52}\right)\\
&=\int_1^{1-i\rho_0} g_+e^{\frac{i}\varepsilon\left(S(1,\tau)-2(\tau-\tau_n)(p-1)\right)}\,dp+
O\left(\varepsilon^{\frac23}z^{-\frac52}\right).
\end{align*}
One proves them using~\eqref{eq:S:p-1-2},~\eqref{eq:fpm} and 
the definitions of $\rho$ and $\rho_0$.  
Similarly one obtains the announced representation for ${\mathcal G}_-$:
first, one ``replaces'' in~\eqref{eq:Fpm}  $\rho$ with $\rho_0$, and 
then, one ``replaces'' $S(p,\tau)$ with $S(1,\tau)-2i(\tau-\tau_n)(p-1)$.
We omit further details.
\end{proof}

Using~(\ref{eq:Gpm:1}) and the definition of $g_-$, 
we obtain
\begin{gather}
  \label{eq:F-}
  {\mathcal G}_-=\frac{c_ne^{i\pi/4}}{\pi^{1/2}}\;{\mathcal G}_-^0+
O\left(\varepsilon^{2/3}/z^{5/2}\right),\\
\label{eq:F-0}
{\mathcal G}_-^0=\int_{1-i\rho_0}^1 e^{-\frac{2i(\tau-\tau_n)(p-1)}\varepsilon}\,\sin(px)\,
\,\frac1{\varepsilon^{\frac32}}\int_1^p(L_0(q)-l_0(q))\,dq\,dp.
\end{gather} 
Let us study  ${\mathcal G}_-^0$ as $\varepsilon\to 0$ and $z\to\infty$. 
As $x$ is bounded, in view of~(\ref{est:A-rough}), we have 
\begin{equation*}
  {\mathcal G}_-^0=\sin x\,\int_{1-i\rho_0}^1 e^{-\frac{2i(\tau-\tau_n)(p-1)}\varepsilon}\,
\,\frac1{\varepsilon^{3/2}}\int_1^p(L_0(q)-l_0(q))\,dq\,dp+
O\left(\varepsilon^{4/3}/z^{2}\right).
\end{equation*}
In view of~(\ref{eq:l0near1}) and~(\ref{eq:L0near1}), we get
\begin{equation*}
  L_0(p)-l_0(p)=\sqrt{2\varepsilon}\,f((p-1)/\varepsilon)+
O(\varepsilon^{3/2}+|p-1|^{3/2}),\quad
f(s)=\zeta(s)+2\sqrt{-s},
\end{equation*}
where the branch of $s\mapsto\sqrt{-s}$ is analytic in $\C$ 
cut along  $\R_+$ and is positive when $s<0$. 
This observation implies that
\begin{equation*}
\begin{split}
  {\mathcal G}_-^0=\varepsilon\sqrt{2}\sin x\,\int_{-i\rho_0/\varepsilon}^0& 
e^{-2i(\tau-\tau_n)t}\,\,\int_0^tf(s)\,ds\,dt\\
&+\varepsilon^2
\int_{-i\rho_0/\varepsilon}^0 e^{-2i(\tau-\tau_n)t}\,O(|t|+|t|^{5/2})\,dt
+O\left(\varepsilon^{4/3}/z^{2}\right).
\end{split}
\end{equation*}
Integrating by parts in the first integral and changing the variable
$t\mapsto (\tau-\tau_n)t$ in the second one, we easily get
\begin{equation*}
\begin{split}
  {\mathcal G}_-^0=\frac{\varepsilon\sqrt{2}\sin x}{2i(\tau-\tau_n)}\,&
\left(e^{-\frac{2\rho_0(\tau-\tau_n)}\varepsilon}\,
\,\int_0^{-i\rho_0/\varepsilon}f(s)\,ds+\int_{-i\rho_0/\varepsilon}^0e^{-2i(\tau-\tau_n)t}\,
\,f(t)dt\right)\\
&+O\left(\varepsilon^{4/3}/z^{2}+
\varepsilon^{5/6}/z^{7/2}\right).
\end{split}
\end{equation*}
In view of~(\ref{eq:zeta-as}), $\int_0^{-i\rho_0/\varepsilon}f(s)\,ds$
is bounded, and, as $\frac{\rho_0(\tau-\tau_n)}\varepsilon
 =z^{\beta}$,
the first term in the brackets is $O(e^{-2z^{\beta}})$. 
Estimate~(\ref{eq:zeta-as}) implies that 
$\int_{-i\infty}^{-i\rho_0/\varepsilon}e^{-2i(\tau-\tau_n)t}g(t)dt=
O(\sqrt{\tau-\tau_n}z^{-3\beta/2}e^{-2z^{\beta}})=
O(e^{-2z^{\beta}})$.
Therefore,
\begin{equation*}
  {\mathcal G}_-^0=\frac{\varepsilon\sqrt{2}\sin x}{2i(\tau-\tau_n)}\,
\operatornamewithlimits\int_{-i\infty}^0e^{-2i(\tau-\tau_n)t}\,f(t)dt+
O\left(\varepsilon^{\frac23}e^{-2z^{2\beta}}+
\varepsilon^{\frac43}z^{-2}+
\varepsilon^{\frac56}z^{-\frac72}\right).
\end{equation*}
Substituting this representation into~(\ref{eq:F-}), we get finally
\begin{equation}\label{eq:last}
 {\mathcal G}_-=\frac{c_ne^{i\pi/4}\varepsilon \sin x}
{\sqrt{2\pi}(\tau-\tau_n)}\,\int_{0}^\infty e^{-2(\tau-\tau_n)t}\,
\left(\zeta(-it)+2e^{i\pi/4}\sqrt{t}\,\right)\,dt
+O\left(\varepsilon^{2/3}/z^{5/2}\right),
\end{equation}
where $\sqrt{t}\ge 0$. The term ${\mathcal G}_+$ is analyzed similarly. It is described 
by~\eqref{eq:last} where  $\zeta(-it)$  is replaced with 
$i\zeta(it)$. In view of~(\ref{eq:zeta}), on the integration path 
$\zeta(-it)=\overline{\zeta(it)}$, and, as 
${\mathcal G}={\mathcal G}_-+{\mathcal G}_+$, we come to~\eqref{as:G_G0}.
\end{proof}
\subsection{Proof of Theorem~\ref{th:final}}
The theorem follows from Proposition~\ref{pro:wedge:main-contrib}
and representations~\eqref{T-T_0},~\eqref{R-R_0} and~\eqref{as:G_G0}.
\bibliographystyle{abbrv}

\begin{thebibliography}{10}
%
\bibitem{Av-El:99}
Avron, J.E., Elgart, A.:
\newblock  Adiabatic theorem without a gap condition.
\newblock Commun. Math. Phys. {\bf 203}, 445-463 (1999)
%
\bibitem{B-L-G:2008}
Babich, V., Lyalinov, M. and Grikurov, V.: 
\newblock {\em Diffraction theory: the Sommerfeld-Malyuzhinets technique.}
\newblock Oxford: Alpha Science, 2008
%
\bibitem{Ch-Fr-Ur}
Chester, C., Friedman, B., Ursell, F.:
\newblock An extension of the method of steepest descents.
\newblock  Proc. Cambridge Philos. Soc. {\bf 53}, 599-611 (1957)
%
\bibitem{Olver}
Olver, Frank W.J.:
\newblock {\em Asymptotics and special functions.}
\newblock New York: Academic Press, 1974
%
\bibitem{F-S:2016}
Smirnov, A.B., Fedotov, A.A.:
\newblock Adiabatic Evolution Generated by a Schr\"odinger Operator 
with Discrete and Continuous Spectra.
\newblock Funct. Anal. Appl. {\bf 50}, 76-79 (2016)
%
\bibitem{Pi:83}
Pierce, A.D.: 
\newblock  Guided mode disappearance during upslope propagation
in variable depth shallow water overlying a fluid bottom. 
\newblock J. Acoust. Soc. Am. {\bf 72}, 523-531 (1982)
%
\bibitem{Wong}
Wong, R.:
\newblock {\em Asymptotic approximations of integrals.}
\newblock Philadelphia: SIAM, 2001
%
\end{thebibliography}

\end{document}